\documentclass[runningheads]{llncs}
\usepackage[T1]{fontenc}
\usepackage{graphicx}
\usepackage[square,sort,comma,numbers]{natbib}
\usepackage[ruled]{algorithm2e} 
\usepackage{framed}
\usepackage{xcolor}
\usepackage{apxproof}
\usepackage{amsmath,amsfonts}

\usepackage[letterpaper, margin=1in]{geometry}

\newcommand\yw[1]{\textcolor{blue}{[Ying: #1]}}

\newcommand{\mechanism}{\mathcal{M}}
\newcommand{\tvd}{\text{TVD}}
\newcommand{\vek}[1]{\boldsymbol{#1}}
\newcommand{\accept}{\mathbf{A}}
\newcommand{\reject}{\mathbf{R}}
\newcommand{\stateL}{{{L}}}
\newcommand{\stateH}{{{H}}}
\newcommand{\signalL}{{{l}}}
\newcommand{\signalH}{{{h}}}
\newcommand{\typeL}{{{\mathcal{L}}}}
\newcommand{\typeH}{{{\mathcal{H}}}}
\newcommand{\Bin}{\mathbf{Bin}} 
\newcommand{\thetaMajority}{\theta_{\texttt{maj}}}
\newcommand{\thetaOurs}{\theta^\ast}
\newcommand{\config}{\{(\alpha_\typeL, \alpha_\typeH), (\mu, P_\signalH^\stateH, P_\signalH^\stateL)\}}

\newtheoremrep{theorem}{Theorem}
\newtheoremrep{lemma}{Lemma}
\newtheoremrep{claim}{Claim}
\newtheoremrep{proposition}{Proposition}
\newtheoremrep{remark}{Remark}

\begin{document}
\title{Aggregation of Antagonistic Contingent Preferences: When Is It Possible?}
%
%
\author{Xiaotie Deng\inst{1}\orcidID{0000-0002-5282-6467} \and
Biaoshuai Tao\inst{2}\orcidID{0000-0003-4098-844X} \and
Ying Wang\inst{1}\orcidID{0009-0005-0834-4531}}
\authorrunning{X. Deng et al.}
%
\institute{
Center on Frontiers of Computing Studies, School of Computer Science, Peking University, Beijing, China 
\and
John Hopcroft Center for Computer Science, School of Electronic Information and Electrical Engineering, Shanghai Jiao Tong University, Shanghai, China
}
\maketitle              

\begin{abstract}
We study a two-alternative voting game where voters' preferences depend on an unobservable world state and each voter receives a private signal correlated to the true world state. We consider the collective decision when voters can collaborate in a group and have antagonistic preferences---given the revealed world state, voters will support different alternatives. 
We identify sharp thresholds for the fraction of the majority-type voters necessary for preference aggregation.

We specifically examine the majority vote mechanism (where each voter has one vote, and the alternative with more votes wins) and pinpoint a critical threshold, denoted as $\thetaMajority$, for the majority-type proportion. When the fraction of majority-type voters surpasses $\thetaMajority$, there is a symmetric strategy for the majority-type that leads to strategic equilibria favoring informed majority decisions. Conversely, when the majority-type proportion falls below $\thetaMajority$, equilibrium does not exist, rendering the aggregation of informed majority decisions impossible.

Additionally, we propose an easy-to-implement mechanism that establishes a lower threshold $\thetaOurs$ (with $\thetaOurs \leq \thetaMajority$) for both equilibria and informed majority decision aggregation. We demonstrate that $\thetaOurs$ is optimal by proving a general impossibility result: if the majority-type proportion is below $\thetaOurs$, with mild assumptions, no mechanism can aggregate the preferences, meaning that no equilibrium leads to the informed majority decision for any mechanism.

\end{abstract}

\section{Introduction}

Consider a virtual community with $n$ residents, a decision about whether to impose a housing tax is to be made. 
The community residents consist of two groups, 
landlords and tenants, with antagonistic interests with respect to the rental prices.
Landlords favor higher rents, while tenants prefer lower rents. 
The community decides by the votes of the residents and a desired decision is the alternative supported by more than half of the residents after the housing tax effect is revealed.
If the tax effect on the rental prices is known ex-ante,  
the desired decision can easily be made, 
for example in the case where the housing tax lowers the rental price. If more than half of the residents are landlords, they reject the housing tax policy. On the other hand, if more than half of the residents are tenants, they will accept the tax policy.

Such a majority vote mechanism provides each resident with two alternatives (either ``accept'' or ``reject'') and designates the one with more votes as the winner, successfully revealing the preference of the majority. Also, this majority vote mechanism is strategy-proof for this case, where the tax effect is known ex-ante.  

However, the policy effect is uncertain before deployment. In practice, it is usually unclear to the residents whether the housing tax will increase or decrease the rental price. 
Instead, each resident independently has a hunch of the policy effect, which is obtained based on public information and individual experience. 
If residents truthfully report their hunches, by the well-known Condorcet Jury Theorem~\cite{de2014essai}, 
when the majority of agents hold a correct hunch, the policy effect can be correctly aggregated with a high probability
via the majority vote.

The real-life situations are even more complex than the idealized scenario described above.
First, when the majority of agents hold an incorrect hunch, the majority vote mechanism fails to reveal the true policy effect.
Second, as is typical in most social choice settings, voters often have differing preferences.
In two alternative cases, the preferences may be antagonistic, as in our housing tax example.
While the majority group aims to jointly determine the true effect and vote for their preferred alternative, the minority group with opposing preferences may try to prevent this through strategic behaviors.

Examples of antagonistic preferences abound in the real world: salary adjustment, two-party presidential election, the tradeoff between freedom and security, and more.  
Often in these cases, the correlation between alternatives and voters' utilities is not immediately clear: 
the impact of a policy is uncertain before deployment, 
the performance of a president-elect remains unpredictable before the inauguration, 
and the consequences of decisions like Brexit were unforeseen by many who abstained from voting against it.
Given these complex scenarios, is it still possible to identify a good collective decision?

Specifically, we study the two alternative voting games where voters/agents' preferences are endogenous and depend on an unobservable world state. In our context, 
\begin{itemize}
    \item Voters are \emph{partially informed}: they receive signals correlated to the world state, and their preferences are contingent, depending on the unknown state.
    \item Preferences are \emph{antagonistic}: when the state is revealed, some voters prefer one alternative, while the remaining voters prefer the other alternative. This naturally gives a partition of the voters into two types: the majority type and the minority type.
\end{itemize}
The objective is to uncover the \emph{informed majority decision}: the alternative favored by more than half of the agents if the world state is revealed.
In our context, the informed majority decision is the alternative favored by the majority-type voters given the revealed world state.

Our work considers \emph{coalitional behaviors}. The two antagonistic preferences lead to two voter groups. 
Voters in a group have the same preference for alternatives and may naturally form a coalition to cooperate for better benefits. This corresponds to the solution concept of \emph{strong Nash equilibrium}, where a strategy profile forms an equilibrium if any subset of agents have no incentives to deviate. While the Nash equilibrium ensures no \emph{individual} voter can benefit from changing strategy, the \emph{strong} Nash equilibrium guarantees that no \emph{group} of voters can jointly deviate to improve their outcomes.


In a sentence, we study the following social choice problem.
\begin{center}
    \textbf{Can good collective decisions be made for partially informed voters \\who have antagonistic preferences and can form coalitions?}
\end{center}

\subsection{Our Contribution}

The two features of voters, \emph{antagonistic preferences} and \emph{coalitional behaviors}, distinguish our work from previous literature.

\paragraph{\textbf{Antagonistic Preferences vs. Homogeneous Preferences.}}
When voters' preferences are homogeneous rather than antagonistic, previous work has shown that good collective decisions can be achieved across different scenarios.
As mentioned, when most agents have a correct hunch on the world state, the informed majority decision can be successfully aggregated based on the Condorcet Jury Theorem~\cite{Condorcet1785:Essai}.
When the majority may be wrong, the surprisingly popular mechanism proposed by Prelec et al.~\cite{prelec2017solution} can ensure correct decisions. 
These methods assume that agents are non-strategic.
In another perspective, when agents can act strategically, issues of incentives do not fail the information aggregation.
Although truthfully revealing information is not an equilibrium~\cite{austen1996information}, the majority vote mechanism (or its variants) can still lead to the desirable collective decision in many settings, both with and without coalitional behaviors~\cite{feddersen1997voting,wit1998rational, mclennan1998consequences, myerson1998extended, duggan2001bayesian, pesendorfer1996swing,feddersen1998convicting,martinelli2002convergence,gerardi2000jury,meirowitz2002informative,coughlan2000defense, han2023wisdom,austen1996information}. 

In contrast, when voters' preferences are antagonistic, the feasibility of information aggregation and the incentives issue become two distinct challenges. Even if there are some strategy profiles for the community to figure out the true world state, minority-type agents are unlikely to follow such strategies.

\paragraph{\textbf{Coalitional Agents vs. Non-Coalitional Agents.}}
Antagonistic preferences have also been explored in previous studies~\cite{acharya2016information,kim2007swing,bhattacharya2013preference, bhattacharya2023condorcet,ali2018adverse}, 
where agents are assumed to be non-coalitional and the solution concept of Nash equilibrium is applied. 
In contrast, our work considers agents' coalitional behaviors with the strong Nash equilibrium concept.

Voters' reasoning for Nash equilibria and strong Nash equilibria is fundamentally different.
For Nash equilibria, a voter reasons his/her strategy as if (s)he is the ``pivotal voter''.
That is, (s)he assumes the remaining voters' votes on the two alternatives are half-half distributed, which is the only case where his/her vote matters.
Conditioning on the half-half vote distribution of the remaining voters, a voter forms a posterior belief on the world state, which eventually affects his/her strategy reasoning.
This pivotal-voter reasoning inherently disregards possible coalitional behaviors of agents. For strong Nash equilibria, on the other hand, the deviating group size plays an important role in the reasoning, as different group sizes have different powers for influencing the outcome of an election.

The following example illustrates the difference between the two equilibrium concepts.

\begin{example}
    Consider the housing tax example mentioned at the beginning, and suppose there are two possible world states $\{L,H\}$ where $L$ indicates the imposition of housing tax lowers the rent and $H$ means that the housing tax makes the rent higher.
    Each resident votes between acceptance or rejection of the housing tax proposal.
    They do not see the world states, and each of them receives a signal from $\{l,h\}$ which is correlated to the world state in the following way:
    \begin{itemize}
        \item if the actual world state is $H$, a resident receives signal $h$ with probability $0.9$ and receives signal $l$ with probability $0.1$;
        \item if the actual world state is $L$, a resident receives signal $h$ with probability $0.4$ and receives signal $l$ with probability $0.6$.
    \end{itemize}
    For simplicity, suppose all the residents are the landlords, who would like to accept the housing tax proposal if the actual world state is $H$ and would like to reject the proposal otherwise, and suppose the number of the residents/landlords is sufficiently large.

    Consider the voting strategy profile where each resident votes for acceptance if receiving signal $h$ and votes for rejection if receiving signal $l$.
    If the actual world state is $H$, about $90\%$ (which is more than $50\%$) of the residents will vote for acceptance, which leads to a good outcome (the house tax proposal will be accepted) for all the residents.
    If the actual world state is $L$, about $40\%$ (which is less than $50\%$) of the residents will vote for acceptance, and the outcome (the house tax proposal will be rejected) is again good for all the residents.
    Therefore, under this strategy profile, the alternative favored by the residents (i.e., the informed majority decision) wins with high probability.

    However, the above strategy profile is not a Bayes Nash equilibrium:
    if the remaining residents vote following this strategy, a resident's best response is to always vote for rejection due to the ``pivotal voter reasoning''.
    To see this, a voter's vote only matters if the remaining voters' votes over the two alternatives are half-half distributed.
    Given that the remaining voters are voting according to their signals and the votes are half-half distributed, it is much more likely the actual world is $L$ (in which case the expected distribution of the votes is $0.4$ versus $0.6$, which is much closer to half-half compared with $0.9$ versus $0.1$ in the case where the actual world is $H$).
    Therefore, if the remaining voters vote according to their signals, the best response is to always vote for rejection instead of voting according to the signal, and the above-mentioned strategy profile is not a Bayes Nash equilibrium.

    When we are considering coalitional agents, we normally focus on groups of voters with non-negligible sizes and disregard the deviation of a single voter (whose behavior is very unlikely to affect the outcome of the election if the number of voters is large).
    The above-mentioned strategy profile is an (approximate) strong Bayes Nash equilibrium.
    This is because the ``correct'' alternative is output with probability approaching $1$ (when the number of voters approaches infinity) under this strategy profile, and no group of voters has an incentive to deviate (given that each voter's ``happiness'' is almost maximized).

    On the other hand, the strategy profile where all residents deterministically vote for acceptance is clearly a Bayes Nash equilibrium, as no resident's deviation can possibly change the outcome (in this case, a voter can never be ``pivotal'', and the pivotal voter analysis is trivial).
    However, this strategy profile is not a strong Bayes Nash equilibrium when talking about coalitional agents.
    Under this strategy profile, all residents do not favor the election outcome when the actual world is $L$.
    Therefore, the group of all residents can deviate and choose the vote-according-to-signal strategy profile instead, which is more beneficial to all of them. \qed
\end{example}

\paragraph{\textbf{Summary.}}
To summarize, when voters have homogeneous preferences, the informed majority decision can be reached even when considering coalitional agents. On the other hand, for antagonistic preferences, the study of coalitional agents is missing from the previous literature.

Does the majority vote mechanism still work?
If not, are there better mechanisms that aggregate agents' preferences and reach the informed majority decision?
In this work, we investigate the aggregation of antagonistic preferences in the presence of coalitional behaviors, which is fundamentally more challenging.

\paragraph{\textbf{Our Results.}} 
We find the aggregation of the informed majority decision is not always possible. 
Deciding informed majority decision in equilibrium becomes impossible for any mechanism when the fractions of the two types of agents (with antagonistic preferences) are close.
This outcome is notably unexpected because, a simple majority vote mechanism can effectively elicit the informed majority decision when preferences are antagonistic but predetermined preferences, or contingent but aligned~\cite{han2023wisdom}.  
We reveal sharp thresholds for the fractions that make the aggregation of antagonistic preference possible.


(1) We first study voters’ strategic behavior in the majority vote mechanism and identify a sharp required threshold $\thetaMajority$ for the proportion of the majority. 
When the fraction of the majority agents surpasses $\thetaMajority$, we show that all $\epsilon$-strong Bayes Nash equilibria are ``good'', i.e., leading to the informed majority decision with a high probability.
We provide a characterization of the symmetric strategy for the majority-type agents to ensure this kind of equilibrium.
For the other cases where the majority proportion is below the threshold, we reveal the nonexistence of the strong Bayes Nash equilibrium. 

(2) We also identify a sharp required threshold $\thetaOurs$ for all ``reasonable'' mechanisms. When the fraction of the majority-type agents surpasses $\thetaOurs$, we present an easy-to-implement mechanism that ensures all equilibria are ``good''. Specifically, truthful reporting forms an equilibrium in this mechanism. On the other hand, when the majority proportion is below $\thetaOurs$, our general impossibility result shows that, under some mild technical assumptions, no mechanism can lead to ``good'' equilibria.

Our truthful mechanism design is not just a direct application of the revelation principle to the majority vote mechanism. 
Notably, our mechanism has a lower threshold, i.e., $\thetaOurs \le \thetaMajority$, with strict inequality in some cases. We explain the distinctions further in Section~\ref{section: connection}. 

\subsection{Related Work}
The information aggregation approach of the social choice problem can be traced back to the well-known Condorcet Jury Theorem~\cite{Condorcet1785:Essai}, which highlights the efficiency of collective decision-making. According to the theorem, in a large election where each voter is more likely to vote correctly than incorrectly, the majority vote mechanism almost surely aggregates the group's common preference. Following this theorem, a series of works in democratic theory have explored collective wisdom in general models~\cite{grofman1983thirteen, miller1986information, boland1989modelling, owen1989proving, berg1993condorcet, ladha1992condorcet, ladha1993condorcet, nitzan1985collective}. 
Austen-Smith and Banks~\cite{austen1996information} are the first to consider the Condorcet Jury Theorem in a game-theoretic framework, revealing that agents might have incentives to deviate from informative voting. Surprisingly, subsequent studies~\cite{feddersen1997voting, wit1998rational, myerson1998extended, mcmurray2017ideology} find that the Condorcet jury theorem's conclusions still hold in equilibrium. An extensive literature has since discussed voting behavior and the effectiveness of the election outcome in various contexts, considering different information environments~\cite{feddersen1997voting, meirowitz2002informative}, uncertain population size~\cite{myerson1998extended}, voluntary election~\cite{pesendorfer1996swing, mcmurray2013aggregating}, the communication of voters~\cite{coughlan2000defense}, and endogenous information setting where agents exert efforts to obtain information~\cite{persico2004committee} or receive public or private information from a politician (persuader)~\cite{alonso2016persuading, chan2019pivotal}. 

\paragraph{\textbf{Conditions for Full Information Equivalence.}}
The aggregation of the informed majority decision, also known as Full Information Equivalence (FIE), is a continuing focus in the literature. 
A body of work examines the conditions for successful aggregation~\cite{mclennan1998consequences, martinelli2002convergence, chakraborty2003efficient, barelli2022full, han2023wisdom}. Barelli et al.~\cite{barelli2022full} highlight that the complexity of information structure can hinder the feasibility of information aggregation. They provide an if-and-only-if condition for information structures to allow a strategy profile that aggregates voters' interests in multi-state and multi-signal settings. Han et al.~\cite{han2023wisdom} consider group deviation. Their work surprisingly reveals that any strategy profile, being a strong equilibrium, can lead to the informed majority decision when contingent voters have aligned preferences. 
Our work also explores how the information structure affects aggregation, taking both coalitional deviation and adversarial preferences into account.
While Barelli et al. provide conditions for some Nash equilibrium to aggregate completely homogeneous preferences, our work identifies conditions for the existence of \emph{strong} Nash equilibrium that can aggregate conflicting preferences. We demonstrate that less informative signals and the closeness of the proportion between two preference groups pose barriers to information aggregation. 
Other sources of aggregation failure include unanimity rules~\cite{feddersen1998convicting}, alternative voters motivations~\cite{razin2003signaling, callander2008majority}, state-dependent number of of voters~\cite{ekmekci2020manipulated}, and so on.

\paragraph{\textbf{Non-common Value Environments.}}
Most studies consider information aggregation with either completely homogeneous preferences~\cite{duggan2001bayesian} or preferences retaining the common values feature~\cite{gerardi2000jury}. However, there are also many works addressing non-common value settings. Acharya considers redistribution in a social mobility model with high and low income voters~\cite{acharya2016information}. Kim and Fey study voluntary election with adversarial voter preferences~\cite{kim2007swing}. Bhattacharya extends the Condorcet Jury Theorem for more general preferences, where each voter's utility is independently drawn from some distribution~\cite{bhattacharya2013preference, bhattacharya2023condorcet}. Ali, Mihm, and Siga study cases where agents' payoffs are negatively correlated~\cite{ali2018adverse}. 
All these works assume agents are non-coalitional and study properties of Nash equilibrium, while ours considers coalitional agents and the strong Nash equilibrium concept. Our preference setting is similar to Bhattacharya's work~\cite{bhattacharya2013preference}. This work offers an if-and-only-if condition on the distribution of voters' utilities to ensure that any symmetric Nash equilibrium makes the informed majority decision prevail.

\paragraph{\textbf{Additional Related Work.}}
Other studies aggregate preferences and information using different techniques. Recent research adopts the “surprisingly popular” answer to aggregate information~\cite{prelec2004bayesian, Hosseini2021, schoenebeck2021wisdom}. Schoenebeck and Tao's mechanism elicits voters' information and preferences using a simple questionnaire and makes a binary decision that aligns with the informed majority decision with a high probability~\cite{schoenebeck2021wisdom}. Prediction markets are also used to aggregate information for decision-making~\cite{hanson1999decision, othman2010decision, monroe2024public}.

\section{Model}
There are two possible world states, $\Omega = \{\stateL, \stateH\}$, where $\stateL$ indicates the imposition of housing tax lowers the rent and $\stateH$ means that the housing tax makes the rent higher. 

Residents are required to vote for one of two alternatives, either accept($\accept$) or reject($\reject$) the housing tax proposal. 
Residents do not know the true world state. 
Instead, each of them has a private judgment, referred to as a random signal. 
This signal, denoted by a random variable $S_i$, takes value from $\mathcal{S} = \{\signalL, \signalH\}$. The signal value $S_i=\signalL$ indicates that resident $i$ believes the world is in state $\stateL$, and similarly, $S_i=\signalH$ implies the belief that state $\stateH$ happens.
We assume the signals that residents receive are identically and independently generated, following a distribution that depends on the actual world state $\omega\in\Omega$. All residents share a common knowledge of the joint distribution between world state and signals.

To streamline the expression, we use $P_{\signalL}^{\stateL}$ to denote the conditional probability $\Pr[S_i = \signalL \mid \omega = \stateL]$, which is the probability that an agent receives signal $\signalL$ given that the world is in state $\stateL$.
Analogously, we define $P_{\signalL}^\stateH$, $P_{\signalH}^\stateL$, and $P_{\signalH}^\stateH$.
By the law of total probability, we have $P_\signalL^\stateH+P_\signalH^\stateH=1$ and $P_\signalL^\stateL+P_\signalH^\stateL=1$.
Lastly, we use $\mu$ to denote the common prior probability that the world is in state $\stateH$.
Correspondingly, $1-\mu$ is the probability that the true world state is $\stateL$.

It is important to note that signals may exhibit asymmetrical favorability 
(one signal, for example, signal $\signalL$, may have a higher probability than another signal in both world states, i.e.,
both $P_\signalL^\stateL > P_\signalH^\stateL$ and  $P_\signalL^\stateH > P_\signalH^\stateH$ hold).
This implies the possibility that most people make mistakes (i.e., the majority of received signals can be inconsistent with the true world state). 
In this work, we consider such biased signals and assume signals are positively correlated with the world states: 
$$P_\signalL^\stateL > P_\signalL^\stateH \text{~and~}P_\signalH^\stateH > P_\signalH^\stateL.$$ 
Let $\Delta = P_\signalL^\stateL- P_\signalL^\stateH = P_\signalH^\stateH - P_\signalH^\stateL$ be the difference in signal frequencies across two world states. Our assumption is equivalent to $\Delta > 0$, implying a general tendency for the signals to align with the true world state.

Each resident $i$ has a utility function $v_i: \Omega \times \{\mathbf{A}, \mathbf{R}\} \to \{0, 1, \ldots, B\}$,  which assigns a payoff value based on the world state and the chosen alternative.\footnote{In many social choice settings, only ordinal preferences are defined for the voters. However, in our setting, cardinal preferences for the voters need to be defined in order to define the equilibrium concepts.} 
In our model, we consider \emph{contingent} agents, whose preferred alternative varies across states. There are two types of agents, i.e., $\mathcal{T} = \{\typeL, \typeH\}$:
\begin{itemize}
    \item An type-$\typeL$ agent (analogous to a tenant in the example) prefers alternative $\accept$ in state $\stateL$ and prefers alternative $\reject$ in state $\stateH$. Formally, $$v_i(\stateL, \accept) > v_i(\stateL, \reject)\text{ and }v_i(\stateH, \reject) > v_i(\stateH, \accept)\text{ for }i\text{ such that }t_i = \typeL.$$
    \item An type-$\typeH$ agent (analogous to a landlord) prefers alternative $\reject$ in state $\stateL$ and prefers $\accept$ in state $\stateH$. Formally, $$v_i(\stateL, \reject) > v_i(\stateL, \accept)\text{ and }v_i(\stateH, \accept) > v_i(\stateH, \reject)\text{ for }i\text{ such that }t_i = \typeH.$$ 
\end{itemize}
As it is standard in the Bayesian game setting, the types of agents and thus their utility functions are private information.

Let $\alpha_\typeL$ denote the approximate fraction of type-$\typeL$ agents, $\alpha_\typeH$ denote that of type-$\typeH$. Since there are only two types of agents, $\alpha_\typeL + \alpha_\typeH = 1$.
Let $\alpha = \max\{\alpha_\typeL, \alpha_\typeH\}$ denote the proportion of the majority-type in the population. We consider cases with condition $\alpha > \frac12$, indicating a clear majority. Formally, given the number of agents $n$, the number of majority-type agents is $\lfloor \alpha\cdot n \rfloor$.

An \emph{instance} of our model includes: 
\begin{itemize}
\item the number of agents, $n$;
\item the common prior of the world state, $(\mu, 1 - \mu)$, where $\mu$ is the probability of state $\stateH$;
\item the signal distributions, $(P_\signalH^\stateH, P_\signalL^\stateH)$ and $(P_\signalH^\stateL, P_\signalL^\stateL)$, where $P_s^\stateH$ and $P_s^\stateL$ denote the conditional probabilities of receiving signal $s \in \mathcal{S}$ in two possible states;
\item the utility functions of agents, $\{v_i\}_{i = 1}^n$ (assumed to be private);
\item and the fraction of each type, $(\alpha_\typeL, \alpha_\typeH)$ (consistent with the utilities).
\end{itemize}

We define a \emph{configuration} of our model as parameters $\{(\alpha_\typeL, \alpha_\typeH), (\mu, P_\signalH^\stateH, P_\signalH^\stateL)\}$. This configuration refers to all instances equipped with the same parameters, where the number of agents $n$ can be arbitrarily large, and the utility functions can be diverse but adhere to the type proportions $(\alpha_\typeH, \alpha_\typeL)$.

Our analyses will delve into agents' voting behaviors and the resulting outcome under different configurations. 
It may be helpful to think of $n \to \infty$ at first and the results for finite $n$ are provided in the appendix.


\subsection{Mechanism, Strategy, and Equilibrium} 
In this section, we will elaborate on the concepts of mechanisms, strategies, and equilibria, which are defined in the standard manner as in game theory.

\paragraph{\textbf{Mechanism.}}
A \emph{mechanism} $\mechanism: \mathcal{R}^n \to \Delta(\{\accept, \reject\})$ is a decision-making process that outputs a distribution over alternatives based on the report profile $(r_1, r_2, \ldots, r_n)\in \mathcal{R}^n$ of voters. The \emph{report space} $\mathcal{R}$ is specified by the questions posed to agents, which can be about their received signals, votes on alternatives, preferences in certain states, or beliefs about others' signals. The mechanism only sees agents' reports and does not know any component of the instance. 


\paragraph{\textbf{Strategy.}}
A (mix) \emph{strategy} for an agent is a function $\sigma: \mathcal{T}\times\mathcal{S} \to \Delta(\mathcal{R})$ that maps the type and the signal received by the agent to a distribution on possible reports. A strategy is said to be \emph{truthful} if it honestly answers all questions based on the agent's knowledge and signal.
When the type of the agents is clear from the context, we slightly abuse the notation and describe a strategy by the function $\sigma:\mathcal{S}\to\Delta(\mathcal{R})$ with the first function input (i.e., the type) omitted.

\paragraph{\textbf{Expected Utility Function.}}
Given a strategy profile $\Sigma = (\sigma_1, \sigma_2, \ldots, \sigma_n)$, we denote the expected utility of agent $i$ under mechanism $\mechanism$ as $u_i^{\mathcal{M}}(\Sigma)$, where the expectation is taken over the sampling of agents’ signals (with agents' strategies,  it decides the report distributions) and the mechanism's outcome. 
The expected utility is concerned with \emph{ex-ante} outcomes, as opposed to the \emph{ex-post} utility $v_i$, which is considered after the world state is realized. In this paper, we focus on \emph{ex-ante} utilities when talking about any equilibrium
solution concept.

\paragraph{\textbf{Equilibrium.}}
In a social choice setting with a large number of agents, it is often the case that any single agent's deviation has a negligible impact on the overall outcome. Therefore, rather than considering the typical Bayes Nash equilibrium, we consider a much stronger concept: the strong Bayes Nash equilibrium, where the collective deviations of subsets of agents are taken into account.

\begin{definition}[$\epsilon$-Strong Bayes Nash Equilibrium]\label{def:BNE}
    A strategy profile $\Sigma = (\sigma_1, \ldots, \sigma_n)$ is an $\epsilon$-strong Bayes Nash equilibrium if no subset of agents $D$ and alternative profile $\Sigma' = (\sigma'_1, \ldots, \sigma'_n)$ exist such that
    \begin{enumerate}
        \item $\sigma_i = \sigma_i'$ for each $i \not \in D$,
        \item $u_i(\Sigma') \ge u_i(\Sigma)$ for each $i \in D$,
        \item there exist $i \in D$ such that $u_i(\Sigma') > u_i(\Sigma) + \epsilon$.
    \end{enumerate}
\end{definition}

\subsection{Objective}
The objective of our work is to unearth the \emph{informed majority decision}, which is the decision favored by the majority if the true world state $\omega$ were known. 
For example, when type-$\typeH$ is the majority type, 
the informed majority decision is $\accept$ if the state is $\stateH$ and $\reject$ otherwise. 
Conversely, when the majority type is type-$\typeL$, the informed majority decision is $\accept$ in state $\stateL$ and $\reject$ otherwise.

We identify the necessary and sufficient conditions on parameters for ensuring the informed majority decision in equilibrium, through both the equilibrium analysis of the majority vote mechanism and the mechanism design view.




\subsection{Majority Vote Mechanism}\label{section: Majority Vote Intro}

The \emph{majority vote mechanism} is a common and natural approach to generating collective decisions. In the majority vote, each voter $i$ votes for either $\accept$ or $\reject$ (the report space $\mathcal{R}=\{\accept,\reject\}$). The alternative supported by the majority will be the output. In the case of a tie, where exactly half of the agents vote for $\accept$, both the alternatives have an equal chance of being selected.

In this context, an agent's (mix) strategy is a mapping from the received signal to a distribution over $\{\accept, \reject\}$. 
We denote such strategy as a vector $(\beta_\signalL, \beta_\signalH)$, where $0 \le \beta_s \le 1$ corresponds to the probability of voting $\accept$ when receiving signal $s \in \mathcal{S}$.

We use alternative notation $(\delta_\signalL, \delta_\signalH)$ to denote the strategy $(\beta_\signalL,\beta_\signalH)$, where $\delta_s \in [-\frac{1}{2}, \frac{1}{2}]$ indicates the deviation of $\beta_s$ from $\frac{1}{2}$. The deviation direction aligns with the possible shift in the preferred alternative of type-$\typeH$ agents when receiving signal $s$. That is, we have 
\[\beta_\signalL = \frac12 - \delta_\signalL\text{ and }\beta_\signalH = \frac12 + \delta_\signalH.\]
We are following using $(\delta_\signalL, \delta_\signalH)$ rather than $(\beta_\signalL, \beta_\signalH)$ to denote strategy for technical simplicity.


Under the majority vote mechanism, it is not clear what strategies are considered truthful: what does it mean by saying that ``reporting $\accept$/$\reject$ truthfully reflects an agent's preference when receiving a signal $\signalL$/$\signalH$''?
Below we provide some candidates for the ``truthful'' strategy in the majority vote to familiarize readers with our context.

A natural strategy that can be considered as ``truthful'' is to report the alternative that exactly corresponds to the signal.
For example, for a type-$\typeH$ agent, this strategy is $(\beta_\signalL,\beta_\signalH)=(0,1)$, i.e., the agent votes for $\accept$ (with probability $1$) when receiving signal $\signalH$ and $\reject$ (with probability $1$) when receiving signal $\signalL$.
This is called \emph{the informative strategy}. In our alternative notation, it can be represented as $(\delta_\signalL, \delta_\signalH) = (\frac12,\frac12)$.
The informative strategy profile does not always lead to the informed majority decision, even when all the agents have the same type, say, type-$\typeH$.
As we have mentioned, the signals may be biased (e.g., it is possible that both $P_\signalH^\stateH$ and $P_\signalH^\stateL$ are more than $\frac12$), voting according to such signals clearly may lead to the wrong outcome.

Another natural ``truthful'' strategy is to cast a vote that maximizes the agent's expected utility, assuming (s)he is the only voter who decides the outcome. An agent, knowing the parameters $\mu, P_\signalH^\stateH$, and $P_\signalH^\stateL$, can update his/her belief of the true world state upon the received private signal in a Bayesian way. Then, (s)he votes for the alternative that maximizes his/her expected utility.
This strategy is called \emph{the sincere strategy}\footnote{The sincere strategy is asymmetric. Voters with the same type and the same received signal may vote differently due to the difference in utility functions.}. 

Both of these ``truthful'' strategy profiles do not lead to the informed majority decision even when all agents have the same type~\cite{han2023wisdom}.
In fact, in order to achieve the informed majority decision, agents need to be more sophisticated than being truthful. Han et al. \cite{han2023wisdom} show that, with common value environment and an infinite number of agents, \emph{a strategy profile leads to the informed majority decision if and only if it is a strong Bayes Nash equilibrium}.
Neither the informative strategy profile nor the sincere strategy profile leads to the informed majority decision, so neither of them is a strong Bayes Nash equilibrium. 
Their work also characterizes the strategy profile that leads to the informed majority decision (or, is a strong Bayes Nash equilibrium) for single-type voters.
We briefly go through these, as many intuitions will be useful to the more general case with two types of agents studied in this paper.

We assume all the agents are of type $\typeH$. According to Han et al.~\cite{han2023wisdom}, the symmetric strategy profile $\{(\delta_\signalL,\delta_\signalH)\}$ aggregates the informed majority decision if it satisfies the following two inequalities: 
\begin{equation}\label{eqn:onetypestrategy}
   \Delta_\accept^\stateH(\delta_\signalL, \delta_\signalH) := P_\signalH^\stateH \cdot\delta_\signalH - P_\signalL^\stateH\cdot\delta_\signalL > 0\text{ and }
    \Delta_\accept^\stateL(\delta_\signalL, \delta_\signalH) := P_\signalH^\stateL \cdot\delta_\signalH - P_\signalL^\stateL\cdot\delta_\signalL < 0.   
\end{equation}

In these inequalities, $\Delta_\accept^\stateH(\delta_\signalL, \delta_\signalH)$ is the expected increment in $\accept$'s vote share in state $\stateH$, compared to the neutral point $\frac12$. The amount $\Delta_\accept^\stateL(\delta_\signalL, \delta_\signalH)$ captures the same for state $\stateL$.

Why does this work? Notice that the difference in expected vote share across two states is $\Delta_\accept^\stateH(\delta_\signalL, \delta_\signalH) -\Delta_\accept^\stateL(\delta_\signalL, \delta_\signalH) = \Delta\cdot (\delta_h + \delta_l)$ (recall that $\Delta=P_\signalH^\stateH-P_\signalH^\stateL = P_\signalL^\stateL - P_\signalL^\stateH$ is the signal frequency difference). The difference in signal frequency, $\Delta$, is utilized to produce the vote share difference, and particularly, the opposite sign between $\Delta_\accept^\stateH(\delta_\signalL, \delta_\signalH)$ and $\Delta_\accept^\stateL(\delta_\signalL, \delta_\signalH)$. 
Such opposite directions ensure the informed majority decision is supported by more than half of the agents in both states. 
In details, given the strategy profile $\{(\delta_\signalL, \delta_\signalH)\}$, the expected vote share of $\accept$ in state $\stateH$ and the vote share of $\reject$ in state $\stateL$ (denoted by $p_\accept^\stateH(\delta_\signalL, \delta_\signalH)$ and $p_\reject^\stateL(\delta_\signalL, \delta_\signalH)$ respectively) can be calculated as
\begin{equation}\label{eqn:onetypep}
 p_\accept^\stateH(\delta_\signalL, \delta_\signalH)
        = \frac12 + \Delta_\accept^\stateH(\delta_\signalL, \delta_\signalH) \text{ and } p_\reject^\stateL(\delta_\signalL, \delta_\signalH)
        = \frac12 -\Delta_\accept^\stateL(\delta_\signalL, \delta_\signalH).
\end{equation}

Therefore, \eqref{eqn:onetypestrategy} implies that the expected vote share of the informed majority decision is more than one half.
By the Law of Large Numbers, the exact vote share concentrates on the expected share as $n\rightarrow\infty$, the informed majority decision is aggregated with a probability approaching $1$. Such a strategy profile is also a strong Bayes Nash equilibrium for common-interest agents.

The above reasoning holds if all agents have the same type $\mathcal{H}$.
The setting with two antagonistic types studied in this paper is much more complex.


\section{Equilibrium Analysis under Majority Vote}

In this section, we study the majority vote mechanism with two types of voters. We identify a sharp threshold, $\thetaMajority$, for the majority proportion. This threshold determines when the majority vote mechanism can consistently aggregate the informed majority decision.

\begin{theorem}\label{thm:MAJ_main}
In the majority vote mechanism, a critical threshold exists for the majority proportion, denoted as $\thetaMajority = \frac1{2M}$ where 
\begin{equation*}
M = \begin{cases}
\frac{P_\signalL^\stateL}{P_\signalL^\stateL + P_\signalL^\stateH}, & \text{if } P_\signalH^\stateL + P_\signalH^\stateH \le 1,  \\
\frac{P_\signalH^\stateH}{P_\signalH^\stateL + P_\signalH^\stateH}, & \text{otherwise}.
\end{cases}
\end{equation*}

For this threshold, there exist functions $\epsilon(n)$, $p(n)$, and $\gamma(n)$ all dependent on the number of agents $n$, such that as $n \to \infty$, $\epsilon(n) \to 0$, $p(n) \to 1$ and $\gamma(n)$ does not approach to $0$. For these functions, the following two statements hold:
\begin{itemize}
    \item If the majority proportion $\alpha$ exceeds this threshold $\thetaMajority$, i.e., $\alpha > \thetaMajority$, an $\epsilon(n)$-strong Bayes Nash equilibrium exists. Furthermore, any such equilibrium leads to the informed majority decision with a probability of at least $p(n)$.
    \item Conversely, if the majority proportion is below this threshold, i.e., $\alpha \le \thetaMajority$, there is no $\gamma(n)$-strong Bayes Nash equilibrium. 
\end{itemize}
\end{theorem}

The remaining part of this section is devoted to proving Theorem~\ref{thm:MAJ_main}.
In Section~\ref{sec:MAJ_optimal_strategy}, we characterize the best strategy for the majority-type agents that is most robust against the possible adversarial actions from the minority-type agents.
In Section~\ref{sec:MAJ_infinite_n}, we prove Theorem~\ref{thm:MAJ_main} assuming $n\rightarrow\infty$, in which case $\epsilon(n)$ and $p(n)$ in the theorem become just $0$ and $1$ respectively.
The assumption $n\rightarrow\infty$ simplifies the proof by avoiding many technical details, and the proof with $n\rightarrow\infty$ already contains many intuitions behind the theorem.
In Section~\ref{sec:MAJ_finite_n}, we prove the original theorem with finite $n$.

\subsection{Optimal Strategy for Majority-type Agents}
\label{sec:MAJ_optimal_strategy}
Without loss of generality, we assume type-$\typeH$ agents are the majority. 

If the majority type-$\typeH$ agents use a strategy that satisfies inequalities in \eqref{eqn:onetypestrategy}, then a natural choice for the type-$\typeL$ agents is a strategy that reverses the directions of the two inequalities in \eqref{eqn:onetypestrategy}. Both types of agents would like to choose strategies that maximize the ``margins'' of vote shares, which are $\Delta_\accept^\stateH(\delta_\signalL, \delta_\signalH)$ and $\Delta_\accept^\stateL(\delta_\signalL, \delta_\signalH)$, to guarantee their favored alternatives.


However, the minority type-$\typeL$ agents are less powerful. They will lose completely if they play in this way, as the margins they create are smaller than the margins created by the majority type-$\typeH$ agents.
To avoid a crushing defeat, type-$\typeL$ agents need to choose a more aggressive strategy---vote for $\accept$ or $\reject$ with probability $1$---to hopefully win in one world state. 
The majority type-$\typeH$ agents, on the other hand, would select $\delta_\signalL$ and $\delta_\signalH$ to maximize both margins $\Delta_\accept^\stateH(\delta_\signalL, \delta_\signalH)$ and $\Delta_\accept^\stateL(\delta_\signalL, \delta_\signalH)$ for a greater chance of victory.

We refer to the symmetric strategy $(\delta_\signalL^\ast,\delta_\signalH^\ast)$ that maximizes the margins for the majority type-$\typeH$ agents as \emph{the optimal strategy}. The following lemma characterizes this strategy.

\begin{lemma}[Characterization of Optimal Strategy]
\label{lem: optimal strategy}
    The strategy $(\delta^*_\signalL,\delta^*_\signalH)$ that maximizes the function $P(\cdot,\cdot)$ is 
\begin{equation}\label{eqn:optimalstrategy}
\left\{
\begin{array}{lcl}
\delta_\signalL^* = \frac{1}{2}\cdot \frac{P_\signalH^\stateL + P_\signalH^\stateH}{P_\signalL^\stateL + P_\signalL^\stateH},~\delta_\signalH^* =\frac{1}{2}  & & {\text{if~} \frac{P_\signalH^\stateL + P_\signalH^\stateH}{2}\le \frac{1}{2},} \\
\delta_\signalL^* = \frac{1}{2},~\delta_\signalH^* =\frac{1}{2}\cdot \frac{P_\signalL^\stateL + P_\signalL^\stateH}{P_\signalH^\stateL + P_\signalH^\stateH} & & {\text{otherwise.}}
\end{array}
\right.
\end{equation}
where $P(\delta_\signalL,\delta_\signalH)=\min\{p_\accept^\stateH(\delta_\signalL,\delta_\signalH),p_\reject^\stateL(\delta_\signalL,\delta_\signalH)\}$ for $p_\accept^\stateH(\delta_\signalL,\delta_\signalH)$ and $p_\reject^\stateL(\delta_\signalL,\delta_\signalH)$ defined in \eqref{eqn:onetypep}.
\end{lemma}

Lemma~\ref{lem: optimal strategy} demonstrates the optimal strategy for the majority type-$\typeH$ agents. This strategy varies depending on the relative frequency of signals $\signalH$ and $\signalL$: if signal $\signalH$ is comparatively rare, agents should always vote for $\accept$ upon receiving signal $\signalH$. However, when they receive signal $\signalL$, they do not always vote for $\reject$, but vote for $\accept$ with a certain probability. For the other case where signal $\signalL$ is rare, agents should always vote for $\reject$ when they receive signal $\signalL$ and vote for $\reject$ with a certain probability when receiving signal $\signalH$.

The following is an example of the optimal strategy and other strategies satisfying inequalities in \eqref{eqn:onetypestrategy}.


\begin{example}[Optimal Strategy and Other Strategies]
    Consider the signal distributions in Table~\ref{tab: signal distribution}, where signal $\signalH$ is generally more common.
\begin{table}[h]
\vspace{-0.5cm}
    \centering
    \begin{tabular}{ccc}
    \hline
         & Signal $\signalH$ & Signal $\signalL$\\
    \hline
        State $\stateH$ & $0.75$ & $0.25$\\
    \hline
        State $\stateL$ & $0.5$ & $0.5$\\
    \hline
    \end{tabular}
    \caption{Signal distributions.}
    \label{tab: signal distribution}
    \vspace{-1cm}
\end{table}

According to inequalities in \eqref{eqn:onetypestrategy}, any strategy $(\delta_\signalL,\delta_\signalH)$ satisfying $\delta_\signalH < \delta_\signalL < 3\delta_\signalH$ will ensure the wish of type-$\typeH$ agents when no type-$\typeL$ agents exist. For example, $(\delta_\signalL,\delta_\signalH)$ can be $(\frac{1}{6}, \frac{1}{8})$. Then agents vote for $\accept$ with a probability of $\frac12 + \frac18 = \frac{5}{8}$ upon signal $\signalH$ and with a probability of $\frac12-\frac16=\frac{1}{3}$ upon signal $\signalL$. Such a strategy makes the expected vote share of $\accept$ in state $\stateH$, and that of $\reject$ in state $\stateL$ be $0.75\times \frac{5}{8} + 0.25\times\frac{1}{3} = \frac{53}{96}$ and $0.5\times\frac{3}{8} +0.5 \times\frac{2}{3} = \frac{25}{48}$, respectively.

On the other hand, the optimal strategy for type-$\typeH$ agents is $(\delta_\signalL^*,\delta_\signalH^*) = (\frac{1}{2}, \frac{3}{10})$. Type-$\typeH$ agents deterministically vote for $\reject$ upon the rarer signal $\signalL$. The expected vote share of $\accept$ in state $\stateH$ and that of $\reject$ in state $\stateL$ are both $\frac{3}{5}$. \qed
\end{example}

In this example, the optimal strategy leads to equal vote shares of the informed majority decision in both world states. This is not a coincidence. The equality of shares is a necessary condition for the optimal strategy. 

\begin{proposition}\label{prop: pAH and pRL are equal}
    The optimal strategy $(\delta_\signalL^*, \delta_\signalH^*)$ satisfies $p_\accept^\stateH(\delta_\signalL^*, \delta_\signalH^*) = p_\reject^\stateL(\delta_\signalL^*, \delta_\signalH^*)$.
\end{proposition}

\begin{proof}
We first refine the possible range for the optimal strategy. Taking arbitrary values of $\delta_\signalL$ and $\delta_\signalH$ with conditions $\delta_\signalL > 0$, $\delta_\signalH > 0$, and $\frac{P_\signalL^\stateH}{P_\signalH^\stateH} < \frac{\delta_\signalH}{\delta_\signalL} < \frac{P_\signalL^\stateL}{P_\signalH^\stateL}$, we have \(p_\accept^\stateH(\delta_\signalL, \delta_\signalH) = \frac{1}{2} + P_\signalH^\stateH \cdot\delta_\signalH - P_\signalL^\stateH\cdot\delta_\signalL > \frac{1}{2}\) and \(p_\reject^\stateL(\delta_\signalL, \delta_\signalH) = \frac{1}{2} + P_\signalL^\stateL\cdot\delta_\signalL - P_\signalH^\stateL\cdot\delta_\signalH > \frac{1}{2}\).
Then the optimal values $\delta_\signalL^*, \delta_\signalH^*$ must satisfy $\min \{p_\accept^\stateH(\delta_\signalL^*, \delta_\signalH^*), p_\reject^\stateL(\delta_\signalL^*, \delta_\signalH^*) \} \ge \min \{p_\accept^\stateH(\delta_\signalL, \delta_\signalH), p_\reject^\stateL(\delta_\signalL, \delta_\signalH) \}>\frac{1}{2}.$
Thus, we can deduce that $\delta_\signalL^*$ and $\delta_\signalH^*$  fall within the range $(0, \frac{1}{2}]$. (If they are outside this range, it is impossible to ensure that both $p_\accept^\stateH(\delta_\signalL^*, \delta_\signalH^*) > \frac{1}{2}$ and $p_\reject^\stateL(\delta_\signalL^*, \delta_\signalH^*) > \frac{1}{2}$.)

Then we show strategy $(\delta_\signalL, \delta_\signalH)$ with $\delta_l, \delta_h \in (0, \frac{1}{2}]$ and $p_\accept^\stateH(\delta_\signalL, \delta_\signalH) \not= p_\reject^\stateL(\delta_\signalL, \delta_\signalH)$ is sub-optimal by providing a better strategy.
\begin{itemize}
    \item If $p_\accept^\stateH(\delta_\signalL, \delta_\signalH) < p_\reject^\stateL(\delta_\signalL, \delta_\signalH)$, given that $\delta_\signalL$ lies in the range $(0,\frac{1}{2}]$, we can adjust $\delta_\signalL$ slightly to create a new strategy $(\delta_\signalL^{+}, \delta_\signalH^{+})$ where $\delta_\signalL^{+} = \delta_\signalL - \epsilon$ (for some small $\epsilon$ such that $0 < \epsilon < \delta_\signalL$), and $\delta_\signalH^{+} = \delta_\signalH$. 
    For a sufficiently small $\epsilon$, it follows that $p_\accept^\stateH(\delta_\signalL, \delta_\signalH) < p_\accept^\stateH(\delta_\signalL^{+}, \delta_\signalH^{+}) < p_\reject^\stateL(\delta_\signalL^{+}, \delta_\signalH^{+}) < p_\reject^\stateL(\delta_\signalL, \delta_\signalH).$ Therefore, the minimum of $p_\accept^\stateH$ and $p_\reject^\stateL$ for $(\delta_\signalL^{+}, \delta_\signalH^{+})$ is greater than that for $(\delta_\signalL^{*}, \delta_\signalH^{*})$, indicating a better strategy.
    \item If $p_\accept^\stateH(\delta_\signalL, \delta_\signalH) > p_\reject^\stateL(\delta_\signalL, \delta_\signalH)$, we adjust $\delta_\signalH$ to $\delta_\signalH^{+} = \delta_\signalH - \epsilon$ (where $0 < \epsilon < \delta_\signalH$) and keep $\delta_\signalL^{+} = \delta_\signalL$. For a sufficiently small $\epsilon$, we get $p_\accept^\stateH(\delta_\signalL, \delta_\signalH) > p_\accept^\stateH(\delta_\signalL^{+}, \delta_\signalH^{+}) > p_\reject^\stateL(\delta_\signalL^{+}, \delta_\signalH^{+}) > p_\reject^\stateL(\delta_\signalL, \delta_\signalH).$ Hence, the minimum of $p_\accept^\stateH$ and $p_\reject^\stateL$ for $(\delta_\signalL^{+}, \delta_\signalH^{+})$ is again greater than that for $(\delta_\signalL, \delta_\signalH)$, proving the superiority of the new strategy.
\end{itemize}
\qedhere
\end{proof}

With this condition, Lemma~\ref{lem: optimal strategy} follows from some calculations.
\begin{proof}[of Lemma~\ref{lem: optimal strategy}]
By Proposition~\ref{prop: pAH and pRL are equal}, we have $$\frac{1}{2} + P_\signalH^\stateH\cdot\delta_\signalH^* - P_\signalL^\stateH\cdot\delta_\signalL^* = \frac{1}{2} + P_\signalL^\stateL\cdot\delta_\signalL^* - P_\signalH^\stateL\cdot\delta_\signalH^*.$$
This leads us to the ratio
$$\frac{\delta_\signalH^*}{\delta_\signalL^*} = \frac{P_\signalL^\stateL + P_\signalL^\stateH}{P_\signalH^\stateL + P_\signalH^\stateH}  = \frac{2 - (P_\signalH^\stateL + P_\signalH^\stateH)}{P_\signalH^\stateL + P_\signalH^\stateH} = \frac{1}{(P_\signalH^\stateL + P_\signalH^\stateH)/2} - 1.$$
So we can express $p_\accept^\stateH$ (also $p_\reject^\stateL$) as
$$\frac{1}{2} + P_\signalH^\stateH\cdot\delta_\signalH^* - P_\signalL^\stateH\cdot\delta_\signalL^* = \frac{1}{2} + P_\signalH^\stateH\cdot\left(\delta_\signalH^* + \delta_\signalL^*\right) - \delta_\signalL^* 
= \frac{1}{2} + \delta_\signalL^*\cdot \left[P_\signalH^\stateH \left(\frac{\delta_\signalH^*}{\delta_\signalL^*} + 1\right)- 1\right],$$ whose value is maximized by selecting the largest possible value for $\delta_\signalL^*$.

Both $\delta_\signalL^*$ and $\delta_\signalH^*$ fall within $(0, \frac{1}{2}]$,  we can determine their optimal values based on the value of $\frac{P_\signalH^\stateL + P_\signalH^\stateH}2$: \begin{itemize}
    \item If $\frac{P_\signalH^\stateL + P_\signalH^\stateH}2 \le \frac{1}{2}$ (indicating that $\delta_\signalH^* \ge \delta_\signalL^*$), the optimal values are $\delta_\signalH^* = \frac{1}{2}$ and $\delta_\signalL^* = \frac{1}{2}\cdot \frac{P_\signalH^\stateL + P_\signalH^\stateH}{P_\signalL^\stateL + P_\signalL^\stateH}$.
    \item If $\frac{P_\signalH^\stateL + P_\signalH^\stateH}2 > \frac{1}{2}$ (indicating that $\delta_\signalH^* < \delta_\signalL^*$), the optimal values are $\delta_\signalL^* = \frac{1}{2}$ and $\delta_\signalH^* = \frac{1}{2}\cdot \frac{P_\signalL^\stateL + P_\signalL^\stateH}{P_\signalH^\stateL + P_\signalH^\stateH}$.
\end{itemize}\qedhere
\end{proof}

\subsection{Equilibrium Analysis for $n \to \infty$}
\label{sec:MAJ_infinite_n}
The previous subsection characterizes the optimal strategy for the majority Type-$\typeH$ agents, which generates the maximum margin for the informed majority decision. Intuitively, when the proportion of minority agents is small, the created margin by the majority is large enough to counter the minority's aggressive strategy, so the informed majority decision will always win. On the other hand, when the minority proportion is relatively large, the margin cannot counter their deterministic voting. Then no equilibrium exists: if the minority agents always vote for $\accept$, then the majority agents will reduce their probability of voting $\accept$; if the majority agents play in this way, the minority agents will vote for $\reject$ instead...


Now we provide a threshold on the majority (or minority) proportion for the existence of the strong Bayes Nash equilibrium. We focus on the cases where $n$ goes to infinity, which captures many ideas and intuitions. 
For cases with large but finite $n$, substantially more analyses are required. We defer them to the next subsection. 

We first derive the vote share of the informed majority decision, among the majority type-$\typeH$ agents, when they use the optimal strategy. 
This value, denoted as $M=P(\delta_\signalL^\ast,\delta_\signalH^\ast)$, can be computed by Equations \eqref{eqn:onetypep} and \eqref{eqn:optimalstrategy}:
\begin{equation}\label{eqn:M}
    M = \left\{
\begin{array}{rcl}
\frac{P_\signalL^\stateL}{P_\signalL^\stateL + P_\signalL^\stateH}& & {\text{if~} \frac{P_\signalH^\stateL + P_\signalH^\stateH}{2} \le \frac{1}{2},} \\
\frac{P_\signalH^\stateH}{P_\signalH^\stateL + P_\signalH^\stateH}& & {\text{otherwise.}}
\end{array}
\right.
\end{equation}

\subsubsection{When the Majority Proportion is Large...} In other words, when the expected vote share of the informed majority decision, cast by the majority type agents, is already greater than one-half ($\alpha M > \frac{1}{2}$), the majority agents ensure the informed majority decision. Since there is no room for the majority type agents to improve utilities, any strategy profile of such a case forms a strong Bayes Nash equilibrium.


\begin{lemma}\label{lem:infinityBNE}
    For configuration $\config$ where the majority proportion $\alpha > \frac1{2M}$ and the number of agents $n \to \infty$, any strategy profile $(\sigma_1, \sigma_2, \ldots, \sigma_n)$ such that the majority type-$\typeH$ agents adopt the optimal strategy forms a strong Bayes Nash equilibrium in the majority vote mechanism. Such equilibria lead to the informed majority decision with a probability of $1$.
\end{lemma}
\begin{proof}
    When $n \to \infty$, the vote share for each alternative concentrates on the expected share. 
    The optimal strategy of the majority type guarantees that at least $M$ fraction of the majority agents' votes support the informed majority decision.
    When $\alpha > \frac{1}{2M}$, no matter what strategy the minority agents take, the fraction of total votes that support the informed majority decision surpasses $\alpha M > \frac{1}{2}$, and the informed majority decision prevails.
    The majority-type agents always get their favored alternative, so their utility will not improve through deviation.
    In addition, none of the majority-type agents can be in the deviation group ($D$ in Definition~\ref{def:BNE}), as any change in the outcome of the voting game will hurt the utility of each majority-type agent.
    On the other hand, the minority-type agents are unable to influence the outcome without the help of the majority-type agents. 
    Therefore, no coalition of agents can benefit by deviating from this strategy profile. 
\end{proof}


Lemma~\ref{lem:infinityBNE} characterized a specific kind of ``good'' equilibria that ensure the informed majority decision. One may ask if there exists some ``bad'' equilibrium. The answer is no. In fact, any strong Bayes Nash equilibrium can ensure the informed majority decision in the majority vote mechanism.

\begin{lemma}\label{lem: BNE to majority wish}
    For the number of agents $n \to \infty$, any strong Bayes Nash equilibrium in the majority vote mechanism leads to the informed majority decision with a probability of $1$.
\end{lemma}

Intuitively, for whatever strategies the minority agents play, the majority agents can ``fine-tune'' the values of $\delta_\signalL^*$ and $\delta_\signalH^*$ to absorb the impact made by the minority. In this way, the majority can improve the outcome if the informed majority decision is not guaranteed to win. Therefore, no ``bad'' equilibrium exists.

\begin{proof}
    We show the majority agents can benefit from deviation if the winning probability of the informed majority decision is less than 1.

    Suppose the minority type-$\typeL$ agents take strategies (either symmetric or asymmetric) such that the expected fraction of $\accept$ votes upon signal $\signalL$ and signal $\signalH$ are $\chi_\signalL^{(\typeL)}$ and $\chi_\signalH^{(\typeL)}$. In this case, the impact of the minority agents is equivalent to that when they adopt a symmetric strategy $(\delta_\signalL, \delta_\signalH)$ with $\frac{1}{2} -\delta_\signalL = \chi_\signalL^{(\typeL)}$ and $\frac{1}{2} + \delta_\signalH = \chi_\signalH^{(\typeL)}$.
    Some type-$\typeH$ agents can conduct a ``counter-strategy'' to absorb the minority type's impact on votes. That is, the symmetric strategy $(-\delta_\signalL, -\delta_\signalH)$ so that $\beta_\signalL = \frac{1}{2} -(- \delta_\signalL) = 1 - \chi_\signalL^{(\typeL)}$ and $\beta_\signalH = \frac{1}{2} + (-\delta_\signalH) = 1 - \chi_\signalH^{(\typeL)}$. The remaining of the majority type-$\typeH$ agents can ensure the dominance of the informed majority decision adopting the optimal strategy $(\delta_\signalL^*, \delta_\signalH^*)$, or any strategy satisfying inequalities in (\ref{eqn:onetypestrategy}). 
    Formally, the majority type-$\typeH$ agents can take strategy $(-\delta_\signalL,-\delta_\signalH)$ with a probability of $\frac{1-\alpha}{\alpha}$ and the optimal strategy $(\delta_\signalL^*, \delta_\signalH^*)$ otherwise.
    This strategy divides the majority into two groups. One group of size $(1-\alpha)\cdot n$ negates the minority's impact, while the other group of size $(2 \alpha - 1)\cdot n$ secures the informed majority decision in both states. The winning probability of the informed majority decision will increase to $1$ when the majority-type agents take the above-mentioned strategy. Therefore, any strategy profile that does not lead to the informed majority decision cannot be a strong Bayes Nash equilibrium. \qedhere
\end{proof}

\subsubsection{When the Majority Proportion is Small...}
On the other hand, when the population of the majority-type agents is not large enough to create a large margin, the informed majority decision is no longer secured.
For whatever strategies the majority agents play, the minority agents can choose either $\accept$ or $\reject$ to vote (deterministically) so that they can win in at least one world state. We find no equilibrium exists in such cases. Consequently, the informed majority decision cannot be aggregated.

\begin{lemma}\label{lem:infinitynoBNE}
For configuration $\config$ where the majority proportion $\alpha \le \frac1{2M}$ and the number of agents $n \to \infty$, no strong Bayes Nash equilibrium exists in the majority vote mechanism.
\end{lemma}
\begin{proof}
If there were a strong Bayes Nash equilibrium, by Lemma~\ref{lem: BNE to majority wish}, the winning probability of informed majority decision equals to $1$. We will show that the minority-type agents can form a coalition and decrease the winning probability in this case, which leads to a contradiction.

Suppose the majority type-$\typeH$ agents take strategies (either symmetric or asymmetric) such that the expected fraction of $\accept$ votes among them upon receiving signal $\signalL$ and signal $\signalH$ are $\chi_\signalL^{(\typeH)}$ and $\chi_\signalH^{(\typeH)}$. In this case, the impact of the majority agents is equivalent to that when they adopt a symmetric strategy $(\delta_\signalL, \delta_\signalH)$ with $\frac{1}{2} -\delta_\signalL = \chi_\signalL^{(\typeH)}$ and $\frac{1}{2} + \delta_\signalH = \chi_\signalH^{(\typeH)}$. We can assume without loss of generality that the majority type-$\typeH$ agents play symmetric strategy $(\delta_\signalL,\delta_\signalH)$.

By Lemma~\ref{lem: optimal strategy}, the minimum vote share of the informed majority decision among the majority-type agents, $P(\delta_\signalL, \delta_\signalH) = \min\{ p_\accept^\stateH(\delta_\signalL, \delta_\signalH),  p_\reject^\stateL(\delta_\signalL, \delta_\signalH)\}$, is at most $P(\delta_\signalL^*, \delta_\signalH^*) = M$. 
Given that $\alpha M \le \frac{1}{2}$, either $\alpha p_\accept^\stateH(\delta_\signalL, \delta_\signalH)\le \frac{1}{2}$ or $\alpha p_\reject^\stateL(\delta_\signalL, \delta_\signalH)\le \frac{1}{2}$ holds. That is, the majority type-$\typeH$ agents' votes are not enough to secure the informed majority decision. For strategy profiles that lead the winning probability of the informed majority decision to be $1$, the minority type-$\typeL$ agents can form a coalition to vote strategically and increase their utilities. If $\alpha p_\accept^\stateH(\delta_\signalL, \delta_\signalH)\le \frac{1}{2}$, type-$\typeL$ agents can deterministically vote for $\reject$ to make the vote share of $\reject$ in total be more than one half in state $\stateH$. If $\alpha p_\reject^\stateL(\delta_\signalL, \delta_\signalH)\le \frac{1}{2}$, type-$\typeL$ agents can deterministically vote for $\accept$ to make $\accept$ win in state $\stateL$.
\end{proof}

The following example applies our results to parameters in Table~\ref{tab: signal distribution}. 

\begin{example}[Threshold for the Existence of Strong Equilibrium]
Within signal distributions given in Table~\ref{tab: signal distribution}, the maximum vote share of the informed majority decision 
cast by the majority type agents is $M = \frac{0.75}{0.5 + 0.75} = \frac35$. 

When the majority proportion, $\alpha$, is greater than $\frac{1}{2M} = \frac56$ (and $n \to \infty$),
it forms a strong Bayes Nash equilibrium and aggregates the informed majority decision as long as the majority type-$\typeH$ agents adopt the optimal strategy.
    
When the majority proportion $\alpha$ is less than or equal to $\frac56$, there is no strong Bayes Nash equilibrium.\qed
        
\end{example}

Lemmas~\ref{lem:infinityBNE}, \ref{lem: BNE to majority wish}, and \ref{lem:infinitynoBNE} prove Theorem~\ref{thm:MAJ_main} for the special case $n\rightarrow\infty$.
\subsection{Equilibrium Analysis for Finite $n$}
\label{sec:MAJ_finite_n}
Our characterizations in Lemmas~\ref{lem:infinityBNE}, \ref{lem: BNE to majority wish}, and \ref{lem:infinitynoBNE} continue to hold for finite $n$ when we consider the approximate strong Bayes Nash equilibrium. 

\begin{lemmarep}\label{lem: strong nash equilibrium for finite n}
For configuration $\config$ where the majority proportion $\alpha > \frac1{2M}$, any strategy profile $\Sigma = (\sigma_1, \ldots, \sigma_n)$ such that the majority type-$\typeH$ agents adopt the optimal strategy forms an $\epsilon$-strong Bayes Nash equilibrium in the majority vote mechanism, where $\epsilon=2B^2 \exp(-2c^2\lfloor\alpha n\rfloor)$ and $c$ is a constant defined by $c = \frac{1}{3}(\alpha M - \frac{1}{2})$.
\end{lemmarep}

The following two propositions, Propositions~\ref{prop: the majority wins with high prob} and \ref{prop: no win-win in majority vote}, are pivotal to validate this theorem.
Proposition~\ref{prop: the majority wins with high prob} says that the optimal strategy in the previous section can still guarantee the victory of majority type-$\typeH$ agents with high probability.
Proposition~\ref{prop: no win-win in majority vote} says that, when the majority type-$\typeH$ agents adopt the optimal strategy, no other strategy profile can (significantly) benefit both a type-$\typeH$ agent and a type-$\typeL$ agent. We defer the proof of these two lemmas and the proof of Lemma~\ref{lem: strong nash equilibrium for finite n} into the appendix.

\begin{propositionrep}[Informed Majority Decision Wins with High Probability]\label{prop: the majority wins with high prob}
When the majority proportion $\alpha > \frac1{2M}$, if the majority type-$\typeH$ agents adopt the optimal strategy in \eqref{eqn:optimalstrategy}, then the informed majority decision wins the majority vote with probability at least $1 - 2\exp(-2c^2\lfloor\alpha n\rfloor)$.
\end{propositionrep}

Proposition~\ref{prop: the majority wins with high prob} holds for the same reasons as before (the margin created by the majority-type agents is large enough to counter any strategies played by the minority agents), with additional application of the Chernoff bound.

\begin{propositionrep}[No Win-Win Scenario in Majority Vote]\label{prop: no win-win in majority vote}
When the majority proportion $\alpha > \frac1{2M}$, 
    consider any strategy profile $\Sigma$ such that the majority type-$\typeH$ agents adopt the optimal strategy in \eqref{eqn:optimalstrategy}. Let $i_1$ be an arbitrary $\typeH$-type agent (i.e., $t_{i_1} = \typeH$) and $i_2$ be an arbitrary $\typeL$-type agent (i.e., $t_{i_2} = \typeL$). Then for any  $\epsilon \ge 2B^2 \exp(-2c^2\lfloor\alpha n\rfloor)$ and any strategy profile $\Sigma'$, neither of the following conditions holds in the majority vote mechanism:
    \begin{enumerate}
          \item[(1)] $u_{i_1}(\Sigma') - u_{i_1}(\Sigma) > \epsilon\text{~and~}u_{i_2}(\Sigma') - u_{i_2}(\Sigma) \ge 0$.
        \item[(2)] $u_{i_2}(\Sigma') - u_{i_2}(\Sigma) > \epsilon\text{~and~}u_{i_1}(\Sigma') - u_{i_1}(\Sigma) \ge 0$.
    \end{enumerate}
\end{propositionrep}

\begin{toappendix}
We first prove Lemma~\ref{lem: strong nash equilibrium for finite n} assuming Proposition~\ref{prop: the majority wins with high prob} and \ref{prop: no win-win in majority vote}.

\begin{proof}[of Lemma~\ref{lem: strong nash equilibrium for finite n}]
    Proposition~\ref{prop: the majority wins with high prob} confirms that the strategy profile $\Sigma$ almost certainly ensures the victory of the informed majority decision. Therefore, a deviation by the majority type agents is not significantly beneficial, as the utility increment would not surpass $2B\exp(-2c^2\lfloor\alpha n\rfloor)$. In other words, for any subset $D$ containing only the majority type agents and any deviation profile $\Sigma'$, there is no $i \in D$ such that $u_i(\Sigma') > u_i(\Sigma) + \epsilon$.

    Proposition~\ref{prop: the majority wins with high prob} also bounds the utility improvement for the minority coalition: irrespective of their strategy, as long as the majority maintains the optimal strategy, the probability of the minority type agents' wish prevailing is capped at $2\exp(-2c^2\lfloor\alpha n\rfloor)$, limiting their utility increment by $2B\exp(-2c^2\lfloor\alpha n\rfloor)$. Formally, for any subset $D$ containing only the minority type agents, any deviation profile $\Sigma' = (\sigma'_1, \cdots, \sigma'_n)$ such that $\sigma'_j = \sigma_j$ for $j \not\in D$, there is no $i \in D$ satisfying $u_i(\Sigma') > u_i(\Sigma) + \epsilon$.

    Moreover, Proposition~\ref{prop: no win-win in majority vote} demonstrates that the cooperation between different types of agents is implausible to result in a substantial joint benefit. This implies that there is no valid coalition $D$ containing both types of agents.

    Therefore, any strategy profile $\Sigma$ such that the majority type-$\typeH$ agents adopt the optimal strategy is an $\epsilon$-strong Bayes Nash equilibrium.
\end{proof}

It now remains to prove the two propositions.

\begin{proof}[of Proposition~\ref{prop: the majority wins with high prob}]
When the majority type-$\typeH$ agents adopt the optimal strategy $(\delta^*_{\signalL}, \delta^*_{\signalH})$, the expected fraction of votes supporting the informed majority decision exceeds $\alpha M$ in both states. By the Chernoff bound, no matter what strategy the minority takes, with probability at least $1 - 2\exp(-2c^2\lfloor\alpha n\rfloor)$, the fraction of votes supporting the informed majority decision by the majority-type agents falls within interval $[\alpha M - c, \alpha M + c]$. Thus, the informed majority decision wins with a probability larger than $1 - 2\exp(-2c^2\lfloor\alpha n\rfloor)$.
\end{proof}

\begin{proof}[of Proposition~\ref{prop: no win-win in majority vote}]
    Without loss of generality, we assume type-$\typeH$ is the majority type. 
    Let $i_1$ be an arbitrary type-$\typeH$ agent, and $i_2$ be an arbitrary type-$\typeL$ agent (i.e., $t_{i_1} = \typeH$ and $t_{i_2} = \typeL$). For simplicity of notations, we shorten their ex-ante and ex-post utility notations as $u_1$, $u_2$, $v_1$, and $v_2$.
    
    By Proposition~\ref{prop: the majority wins with high prob}, the probability of the majority vote mechanism failing to satisfy the informed majority decision is less than $2\exp(-2c^2\lfloor\alpha n\rfloor)$. Thus, any potential increment in utility for a type-$\typeH$ agent will not surpass $2B\exp(-2c^2\lfloor\alpha n\rfloor)$, which is less than $\epsilon$, implying that Condition (1) is never met.

    We are next to show Condition (2) is also not feasible. 

    As we can see from Equation~(\ref{eq: typeH ex-ante utility diff}) and Equation~(\ref{eq: typeL ex-ante utility diff}), for both agents to benefit from the deviation (i.e., $\Delta u_{1}(\Sigma, \Sigma') \ge 0$ and $\Delta u_{2}(\Sigma, \Sigma') \ge 0$), the term $(\lambda^\stateH_\accept(\Sigma') -\lambda^\stateH_\accept(\Sigma))$ and $(\lambda^\stateL_\accept(\Sigma') -\lambda^\stateL_\accept(\Sigma))$ must either both be non-negative or both be non-positive.

    Intuitively, for agent $i_1$ of type-$\typeH$, the deviation should result in increasing utility in one state but decreasing utility in the other. Since the optimal strategy of the majority type-$\typeH$ agents already ensures that the informed majority decision is fulfilled with a probability close to $1$, any increment of utility for agent $i_1$ is relatively small. Consequently, the decrement of utility in the other state is also small. Such deviations do not significantly alter the outcome in both states, so the impact on the utility of type-$\typeL$ agent $i_2$ is also minor.
    
    Formally, we consider two cases based on the value $\lambda^\stateH_\accept(\Sigma') - \lambda^\stateH_\accept(\Sigma)$: 
    \begin{itemize}
        \item If $\lambda^\stateH_\accept(\Sigma') - \lambda^\stateH_\accept(\Sigma) \ge 0$, it holds that $\lambda^\stateL_\accept(\Sigma') - \lambda^\stateL_\accept(\Sigma) \ge 0$ (otherwise type-$\typeL$ agents' utilities decrease). 
        We say the probability of majority vote mechanism outputting $\accept$ in state $\stateH$ will not increase a lot, i.e., $\mu(\lambda^\stateH_\accept(\Sigma') - \lambda^\stateH_\accept(\Sigma)) \le 2B\exp(-2c^2\lfloor\alpha n\rfloor)$, because the term $\mu(\lambda^\stateH_\accept(\Sigma') - \lambda^\stateH_\accept(\Sigma))$ is bounded by $\mu(1 - \lambda^\stateH_\accept(\Sigma))$, whose value, by Proposition~\ref{prop: the majority wins with high prob}, is less than $2\exp(-2c^2\lfloor\alpha n\rfloor)$. 
        Notice that $1 \le \Delta v_1^\stateL, \Delta v_2^\stateH \le B$ hold for both types of agents. According to Equation (\ref{eq: typeH ex-ante utility diff}), for type-$\typeH$ agent $i_1$ to benefit, the change amount of utility in state $\stateL$, represented as $(1-\mu)(\lambda^\stateL_\accept(\Sigma') -\lambda^\stateL_\accept(\Sigma))\Delta v_{1}^\stateL$, must not exceed $2B\exp(-2c^2\lfloor\alpha n\rfloor)$. Substituting the inequality relationship into Equation~(\ref{eq: typeL ex-ante utility diff}), we can derive that the ex-ante utility difference for the type-$\typeL$ agent $i_2$ will be less than $2B^2 \exp(-2c^2\lfloor\alpha n\rfloor)$.
        \begin{align*}
            \Delta u_{2}(\Sigma, \Sigma') &< (1-\mu)(\lambda^\stateL_\accept(\Sigma') -\lambda^\stateL_\accept(\Sigma))\Delta v_{2}^\stateL\\
            &= (1-\mu)(\lambda^\stateL_\accept(\Sigma') -\lambda^\stateL_\accept(\Sigma))\Delta v_{1}^\stateL\cdot \frac{\Delta v_{2}^\stateL}{\Delta v_{1}^\stateL}\\
            &\le 2B \exp(-2c^2\lfloor\alpha n\rfloor) \cdot B\\
            & = 2B^2 \exp(-2c^2\lfloor\alpha n\rfloor).
        \end{align*}

        \item If $\lambda^\stateH_\accept(\Sigma') -\lambda^\stateH_\accept(\Sigma) < 0$, it holds that $\lambda^\stateL_\accept(\Sigma')- \lambda^\stateL_\accept(\Sigma) < 0$. Similar to the previous analysis, the utility increment of type-$\typeH$ agent $i_1$ in state $\stateL$ is small: $$(1-\mu)|\lambda^\stateL_\accept(\Sigma') -\lambda^\stateL_\accept(\Sigma)|\cdot \Delta v_{1}^\stateL < 2B\exp(-2c^2\lfloor\alpha n\rfloor).$$
        To prevent the overall ex-ante utility of type-$\typeH$ agent $i_1$ from decreasing, it must hold that $$-2B\exp(-2c^2\lfloor\alpha n\rfloor) < \mu(\lambda^\stateH_\accept(\Sigma') - \lambda^\stateH_\accept(\Sigma))\Delta v_{1}^\stateH < 0.$$
        Consequently, for a type-$\typeL$ agent $i_2$, his/her utility increment is small:
        $$\Delta u_{2} < -\mu(\lambda^\stateH_\accept(\Sigma') - \lambda^\stateH_\accept(\Sigma))\Delta v_{2}^\stateH < 2B^2\exp(-2c^2\lfloor\alpha n\rfloor) \le \epsilon.$$
    \end{itemize}
    In both cases, Condition (2) is not satisfied.
\end{proof}    
\end{toappendix}

Proposition~\ref{prop: no win-win in majority vote} is non-trivial: while agents may exhibit antagonistic preferences in \emph{ex-post}, this antagonism may not appear in \emph{ex-ante}. Agents' different sensitivities in the two states contribute to this complexity. An illustrative example is provided in Example~\ref{example: ex-ante utility}.

To explain the example, we define some notations.
Let $\lambda^\stateL_\accept(\Sigma)$ be the probability that the mechanism outputs alternative $\accept$ in state $\stateL$ under agents' strategy profile $\Sigma$. 
Similarly, $\lambda^\stateH_\accept(\Sigma)$ refers to the corresponding probability in state $\stateH$.
Denote the ex-ante utility difference $u_{i}(\Sigma') - u_{i}(\Sigma)$ for agent $i$ in shifting from profile $\Sigma$ to profile $\Sigma'$ as $\Delta u_i(\Sigma, \Sigma')$. 
Also, we denote the ex-post utility difference in state $\stateH$ for $i$ as $\Delta v_{i}^\stateH = |v_{i}(\stateH, \accept) - v_{i}(\stateH, \reject)|$, and similarly for state $\stateL$.
The following two equations demonstrate the relationship between ex-ante and ex-post utility for both types of agents.
\begin{itemize}
    \item For agent $i$ with type $t_i = \typeH$,
    \begin{align}\label{eq: typeH ex-ante utility diff}
        \Delta u_{i}(\Sigma, \Sigma') = \mu\left(\lambda^\stateH_\accept(\Sigma') - \lambda^\stateH_\accept(\Sigma)\right)\Delta v_{i}^\stateH - (1-\mu)\left(\lambda^\stateL_\accept(\Sigma') -\lambda^\stateL_\accept(\Sigma)\right)\Delta v_{i}^\stateL.
    \end{align}
    \item For agent $i$ with type $t_i = \typeL$,
    \begin{equation}\label{eq: typeL ex-ante utility diff}
             \Delta u_{i}(\Sigma, \Sigma') = (1-\mu)\left(\lambda^\stateL_\accept(\Sigma') -\lambda^\stateL_\accept(\Sigma)\right)\Delta v_{i}^\stateL - \mu\left(\lambda^\stateH_\accept(\Sigma') - \lambda^\stateH_\accept(\Sigma)\right)\Delta v_{i}^\stateH. 
    \end{equation}
\end{itemize}

\begin{example}[Agreement of Ex-Ante Benefit]\label{example: ex-ante utility}
Consider a scenario where both states are equally probable (i.e., $\mu = \frac{1}{2}$). The signal distribution is shown in Table~\ref{tab: signal distribution}. We have three agents, two of type-$\typeH$ and one of type-$\typeL$. Their utilities are detailed in Table~\ref{tab: H-type utility} and Table~\ref{tab: L-type utility}.

\begin{table}[h]
    \centering
    \begin{minipage}{.45\linewidth}
        \centering
        \begin{tabular}{ccc}
        \hline
             &  Accept & Reject\\
        \hline
             State $\stateH$ &  1 & 0\\
        \hline
             State $\stateL$ & 0 & 5\\
        \hline
        \end{tabular}
        \caption{Utility of Type-$\typeH$ Agents}
        \label{tab: H-type utility}
    \end{minipage}%
    \begin{minipage}{.45\linewidth}
        \centering
        \begin{tabular}{ccc}
        \hline
             &  Accept & Reject\\
        \hline
             State $\stateH$ & 0 & 5\\
        \hline
             State $\stateL$ & 1 & 0\\
        \hline
        \end{tabular}
        \caption{Utility of Type-$\typeL$ Agents}
        \label{tab: L-type utility}
    \end{minipage}
\end{table}

    Consider the following two strategy profiles:
    \begin{itemize}
        \item Strategy Profile $\Sigma$: All agents vote for $\accept$ on signal $\signalH$, and $\reject$ on signal $\signalL$.
        
        \item Strategy Profile $\Sigma'$: All agents adopt an uninformative strategy. They vote for $\accept$ with a probability of $\frac13$ and vote for $\reject$ otherwise.
    \end{itemize}
    
    Under strategy profile $\Sigma$, the probability of the majority vote mechanism outputting $\accept$ in state $\stateL$ is $\lambda_\accept^\stateL(\Sigma) = \frac{1}{2}$ and the probability of outputting $\accept$ in state $\stateH$ is $\lambda_\accept^\stateH(\Sigma) = \frac{27}{32}$. 

    Under strategy profile $\Sigma'$, the probability of the majority vote mechanism outputting $\accept$ in both states is $\lambda_\accept^\stateH(\Sigma') = \lambda_\accept^\stateH(\Sigma') = \frac{7}{27}$.

    In the shifting from $\Sigma$ to $\Sigma'$, the change in the probability of outputting $\accept$ in state $\stateL$ is $(1-\mu)(\lambda^\stateL_\accept(\Sigma') - \lambda^\stateL_\accept(\Sigma)) = \frac{1}{2}\times(\frac{7}{27} - \frac{1}{2})\approx -0.12$.
    That change in state $\stateH$ is $\mu(\lambda^\stateH_\accept(\Sigma') - \lambda^\stateH_\accept(\Sigma)) = \frac{1}{2}\times(\frac{7}{27} - \frac{27}{32})\approx -0.29$.
    
For an agent $i_1$ of type-$\typeH$, by Equation~(\ref{eq: typeH ex-ante utility diff}), this shift yields an ex-ante utility difference of $\Delta u_{i_1}(\Sigma, \Sigma') \approx -0.29 \times 1 - (-0.12)\times 5 = 0.31 > 0$.
For an agent $i_2$ of type-$\typeL$, by Equation~(\ref{eq: typeL ex-ante utility diff}), the ex-ante utility difference is $\Delta u_{i_2}(\Sigma, \Sigma') \approx -0.12\times1 - (-0.29) \times 5 = 1.33 > 0$.
Both types of agents find it beneficial to accept the change in strategy, despite their antagonistic preferences. \qed
\end{example}

\begin{lemmarep}\label{lem: BNE2majoritywish}
   Any $\epsilon$-strong Bayes Nash equilibrium in the majority vote mechanism leads to the informed majority decision with a probability of at least $1 - 2\epsilon$, where $\epsilon = B\exp(-2c^2n)$ and constant $c$ is defined as $\frac{1}{2}(2\alpha - 1)(M - \frac{1}{2})$.
\end{lemmarep}

The proof idea of Lemma~\ref{lem: BNE2majoritywish} is analogous to Lemma~\ref{lem: BNE to majority wish}, with Hoeffding’s inequality applied to bound the probability of bad cases.

\begin{proof}
    For contradiction, we assume the probability of the informed majority decision failing is more than $2B\exp(-2c^2n)$, i.e., $\mu\lambda^\stateH_\reject(\Sigma) + (1-\mu)\lambda^\stateL_\accept(\Sigma) \ge 2B\exp(-2c^2 n)$. We challenge this by showing that type-$\typeH$ agents can significantly improve their utility. 

    Without loss of generality, we suppose the minority type-$\typeL$ agents take a symmetric strategy $(\delta_\signalL, \delta_\signalH)$. In this case, the majority type-$\typeH$ agents can take a mirror strategy, i.e., $(-\delta_\signalL, -\delta_\signalH)$ with probability $\frac{1-\alpha}{\alpha}$, and the optimal strategy $(\delta_\signalL^*, \delta_\signalH^*)$ for the rest of the time to almost dominate the game. In this way, the expected fraction of votes for the informed majority decision in both states will be $(2\alpha - 1)(M - \frac{1}{2}) + \frac{1}{2}$.
         Using Hoeffding's inequality, we can see that the probability of the informed majority decision failing is less than $\exp(-2c^2 n)$ in both states. By Equation~(\ref{eq: typeH utility diff}), this deviation leads to an increment in type-$\typeH$ agents' utilities of at least $2B\exp(-2c^2 n) - B\exp(-2c^2 n) = B \exp(-2c^2 n) = \epsilon$, implying that type-$\typeH$ agents have a strong incentive to deviate from $\Sigma$. Therefore, if a strategy profile $\Sigma$ fails the informed majority decision with probability more than $2\epsilon$, such a strategy profile $\Sigma$ is not an $\epsilon$-strong Bayes Nash equilibrium.
\end{proof}

\begin{lemmarep}\label{lem: no strong BNE for finite n}
For configuration $\config$ where the majority proportion $\alpha \le \frac1{2M}$, no $\epsilon$-strong Bayes Nash equilibrium exists in the majority vote mechanism. Here, $\epsilon$ is defined as $\epsilon = \frac{1}{4}\min\{\mu, 1-\mu\} - 2B^2\exp(-2c^2n)$ and constant $c$ is $\frac{1}{2}(2\alpha - 1)\min\{M - \frac{1}{2}, \frac{1}{2} - \alpha M\}$.
\end{lemmarep}
\begin{proof}
     Still, suppose type-$\typeH$ is the majority type. Without loss of generality, we consider the case that agents with the same type use symmetric strategies. If agents with the same type take asymmetric strategies, there exists a symmetric strategy imposing the same effect on the votes.
    
   To validate this lemma, we rewrite agents' ex-ante utility difference $\Delta u_t(\Sigma, \Sigma') = u_t(\Sigma') - u_t(\Sigma)$ using the probability that the informed majority decision fails in both states (i.e., the probability that the majority vote outputs alternative $\reject$ in state $\stateH$, denoted as $\lambda_\reject^\stateH$, and the probability that alternative $\accept$ wins in state $\stateL$, denoted as $\lambda_\accept^\stateL$). The relationship of these amounts is demonstrated as follows:
   \begin{itemize}
    \item For agent $i$ with type $t_i = \typeH$,
    \begin{align}
        \Delta u_{i}(\Sigma, \Sigma') = [\mu\lambda^\stateH_\reject(\Sigma)\Delta v_{i}^\stateH + (1-\mu)\lambda^\stateL_\accept(\Sigma)\Delta v_{i}^\stateL] - [\mu\lambda^\stateH_\reject(\Sigma')\Delta v_{i}^\stateH + (1-\mu)\lambda^\stateL_\accept(\Sigma')\Delta v_{i}^\stateL].\label{eq: typeH utility diff}
    \end{align}
    
    \item For agent $i$ with type $t_i = \typeL$,
    \begin{align}
        \Delta u_{i}(\Sigma, \Sigma') = [\mu\lambda^\stateH_\reject(\Sigma')\Delta v_{i}^\stateH + (1-\mu)\lambda^\stateL_\accept(\Sigma')\Delta v_{i}^\stateL] - [\mu\lambda^\stateH_\reject(\Sigma)\Delta v_{i}^\stateH + (1-\mu)\lambda^\stateL_\accept(\Sigma)\Delta v_{i}^\stateL].\label{eq: typeL utility diff}
    \end{align}
   \end{itemize}
    We show any strategy profile $\Sigma$ cannot form an $\epsilon$-strong Bayes Nash equilibrium based on the value of $\mu\lambda^\stateH_\reject(\Sigma) + (1-\mu)\lambda^\stateL_\accept(\Sigma)$.
    \begin{itemize}
        \item If $\mu\lambda^\stateH_\reject(\Sigma) + (1-\mu)\lambda^\stateL_\accept(\Sigma) \ge \frac{1}{4}\min\{\mu, 1-\mu\}$, type-$\typeH$ agents can significantly improve their utility. By adopting a strategy that mirrors the type-$\typeL$ agents' strategy (i.e., take strategy $(-\delta_\signalL, -\delta_\signalH)$ if type-$\typeL$ agents take symmetric strategy $(\delta_\signalL, \delta_\signalH)$) with probability $(1-\alpha)/\alpha$, and adopt the optimal strategy $(\delta_\signalL^*, \delta_\signalH^*)$ for the rest of the time, the expected fraction of votes for the informed majority decision in both states will be $(2\alpha - 1)(M - \frac{1}{2}) + \frac{1}{2}$.
         Using Hoeffding's inequality, we can see that the probability of the informed majority decision failing is less than $\exp(-2c^2 n)$ in both states. By Equation~(\ref{eq: typeH utility diff}), this strategy leads to an increment in utility of at least $\frac{1}{4}\min\{\mu, 1-\mu\} - B\exp(-2c^2 n) = \epsilon$, indicating that type-$\typeH$ agents have a strong incentive to deviate from $\Sigma$.
        \item If $\mu\lambda^\stateH_\reject(\Sigma) + (1-\mu)\lambda^\stateL_\accept(\Sigma) < \frac{1}{4}\min\{\mu, 1-\mu\}$, we show type-$\typeL$ agents can benefit from deviating. Given that $\alpha M \le \frac{1}{2}$, no matter what strategy $(\delta_\signalL, \delta_\signalH)$ type-$\typeH$ agents adopt, either $\alpha p_\accept^\stateH(\delta_\signalL, \delta_\signalH)\le \frac{1}{2}$ or $\alpha p_\reject^\stateL(\delta_\signalL, \delta_\signalH)\le \frac{1}{2}$. In both scenarios, type-$\typeL$ agents can strategically vote to reduce the approval rate of informed majority decision (by always voting for $\accept$ if $\alpha p_\accept^\stateH(\delta_\signalL, \delta_\signalH)\le \frac{1}{2}$ or always voting for $\reject$ when $\alpha p_\reject^\stateL(\delta_\signalL, \delta_\signalH)\le \frac{1}{2}$). Again using the Hoeffding’s inequality, we can see that there exists one state such that the minority agents' wish wins for more than half of the chance, thus $\mu\lambda^\stateH_\reject(\Sigma') + (1-\mu)\lambda^\stateL_\accept(\Sigma') \ge \frac{1}{2}\min\{\mu, 1-\mu\}$. Consequently, by Equation~(\ref{eq: typeL utility diff}) the utility increment of type-$\typeL$ agents is at least $\frac{1}{2}\min\{\mu, 1-\mu\} - \frac{1}{4}\min\{\mu, 1-\mu\} = \frac{1}{4}\min\{\mu, 1-\mu\}$, which is above $\epsilon$.
    \end{itemize}
    In summary, we can conclude that no strategy profile $\Sigma$ can form an $\epsilon$-strong Bayes Nash Equilibrium in this setting. 
\end{proof}

Intuitively, in this case, where the group sizes of both types are relatively balanced, no matter to which extent the informed majority decision wins, at least one group of agents will deviate for benefit. If the informed majority decision fails with a non-negligible probability, the majority-type agents will collude to improve the outcome. On the other hand, if the informed majority decision almost always wins, the minority-type agents can find a strategy to vote against the majority agents and make their favored alternative win in one state.

We defer the proof of Lemmas~\ref{lem: BNE2majoritywish} and \ref{lem: no strong BNE for finite n} to the appendix.

\section{Results of Truthful Mechanism Design}

In this section, we study the threshold for the majority proportion to ensure the informed majority decision in general \emph{anonymous} mechanisms, where we say that a mechanism is \emph{anonymous} if it always outputs the same (distribution of) alternative(s) for any report profiles $\vek{r}_1,\vek{r}_2\in \mathcal{R}^n$ such that $\vek{r}_1$ is a permutation of $\vek{r}_2$. We design an easy-to-employ mechanism and establish a more attainable threshold $\theta^\ast$.
This mechanism has the lowest requirement for the majority proportion: when $\alpha$ is below $\theta^\ast$, the informed majority decision cannot always be extracted in any anonymous mechanism, even as $n \to \infty$. 



\begin{theorem}\label{thm:ourMechanism}
There exists an anonymous mechanism $\mechanism$ with a critical threshold $\thetaOurs = \frac1{\Delta + 1}$ for the majority proportion such that
\begin{itemize}
    \item If majority proportion $\alpha$ exceeds this threshold $\thetaOurs$, i.e., $\alpha > \thetaOurs$, then 
    truthful reporting becomes an $\epsilon(n)$-strong Bayes Nash equilibrium and this equilibrium ensures the informed majority decision with probability at least $p(n)$, where $\epsilon(n) \to 0$ and $p(n) \to 1$ as $n \to \infty$.
    
\end{itemize}
\end{theorem}

\begin{remark}
Our proposed mechanism provides a solution that aggregates the informed majority decision when agents truthfully report. However, there is no guarantee of agents' truthful reporting when $\alpha \le \thetaOurs$.  
Our impossibility results in Section~\ref{section: impossibility results} show that, in such cases, under some mild assumptions, any anonymous mechanism cannot lead to an $\epsilon$-strong Bayes Nash equilibrium that aggregates the informed majority decision. 
This complements the positive result in Theorem~\ref{thm:ourMechanism}.
\end{remark}


The threshold of our mechanism $\thetaOurs$ is generally more attainable than the threshold of the majority vote mechanism $\thetaMajority$. In other words, the inequality $\frac1{2M} \ge \frac1{\Delta + 1}$ always holds with the equality occurring when $P_\signalH^\stateL + P_\signalH^\stateH = 1$. The proof of this finding is provided below.

\begin{proposition}\label{obs: threshold ineq}

  For any configuration $\config$, it holds that $\thetaMajority \ge \thetaOurs$, i.e., $\frac1{2M} \ge \frac1{\Delta + 1}$. Particularly, the equality holds when $ P_\signalH^\stateL + P_\signalH^\stateH = 1$.
\end{proposition}

\begin{proof}

It suffices to show $2M \le \Delta +1$.

Let $k$ denote the average occurrence rate of signal $\signalH$ across both worlds, i.e., $k = \frac{1}{2}( P_\signalH^\stateL + P_\signalH^\stateH)$. Using $\Delta$ and $k$, we can rewrite the signal frequencies as: 
\begin{align*}
    &P_\signalH^\stateH = k + \Delta/2, &P_\signalL^\stateH = 1-k - \Delta/2,\\
    &P_\signalH^\stateL = k - \Delta/2, &P_\signalL^\stateL = 1-k + \Delta/2.
\end{align*}

By Equation~(\ref{eqn:M}), the value of $M$ is $\frac{1-k + \Delta/2}{2(1-k)}$ for $k \le \frac{1}{2}$ and $\frac{k + \Delta/2}{2k}$ for $k > \frac{1}{2}$.

Thus, we know that $2M$ is either $1 + \frac{\Delta}{2(1 - k)}$ (for $k \le \frac{1}{2}$) or $1 + \frac{\Delta}{2k}$ (for $k > \frac{1}{2}$). In both cases, the inequality $2M \le \Delta + 1$ holds, and the equation holds when $k = \frac{1}{2}$.
\end{proof}

\subsection{Our Mechanism}




Our mechanism draws inspiration from Schoenebeck and Tao's mechanism~\cite{schoenebeck2021wisdom}. This work extends Prelec et al.'s surprisingly popular algorithm~\cite{prelec2017solution} to a truthful mechanism, where agents' posterior estimations of the signal frequencies serve as a threshold to distinguish two world states. 
We follow this idea while a nuanced threshold value is elicited. 

\begin{framed}
 \noindent\textbf{Our Mechanism}\\
 \rule{\textwidth}{0.5pt}
\begin{enumerate}
\item \textbf{Information Collection}.\\
Each agent $i$ reports his/her private type ($\typeL$ or $\typeH$),  the received signal ($\signalL$ or $\signalH$), and a threshold value $\hat{\delta}_i$ defined as $\frac{P_{\signalL}^\stateL}{P_{\signalL}^\stateL + P_{\signalH}^\stateH}$.

\item \textbf{Majority Type Identification}.\\
Determine the majority-type, $\hat{M}$, from the agents' reports.
\item \textbf{Threshold Calculation}.\\
Calculate the collective threshold, $\hat{\delta}$, as the median of the reported $\hat{\delta}_i$ values.
\item \textbf{State Assessment}. \\
 Determine the true world state, $\hat{\omega}$, based on whether the reporting frequency of $\signalL$ exceeds the threshold $\hat{\delta}$. Set $\hat{\omega} = \stateL$ if the threshold is exceeded; $\hat{\omega} = \stateH$ otherwise.

\item \textbf{Outcome Selection}. \\Choose the preferred alternative of majority-type $\hat{M}$ in world state $\hat{\omega}$.
\end{enumerate}
\end{framed}

This mechanism elicits a threshold value $\hat{\delta}$ (expected to be $\delta = \frac{P_{\signalL}^\stateL}{P_{\signalL}^\stateL + P_{\signalH}^\stateH}$). Such a threshold value may seem artificial and challenging to estimate, as it demands that agents have a precise understanding of the parameters of the information structure. However, we can obtain this value in a user-friendly way by asking agents for their posterior forecasts on the world state (e.g., $\Pr[\omega = \stateL \mid S_i = \signalL]]$) and the estimations of signals received by peers (e.g., $\Pr[S_j = \signalL \mid s_i = \signalH]$). The details are referred to in Appendix~\ref{sect: threshold elicitation}.

Below is a running example of our mechanism and Schoenebeck and Tao's mechanism~\cite{schoenebeck2021wisdom}. As is shown, ours provides a more attainable threshold on the majority proportion to guarantee the informed majority decision under minority agents' group strategic behavior.

\begin{example}[Running Example of Mechanisms]
Given the signal distribution in Table~\ref{tab: signal distribution} and a prior probability of state $\mu = \frac{1}{2}$, the thresholds for our mechanism and Schoenebeck and Tao's mechanism are calculated as follows.

If all the majority agents report truthfully, the threshold $\delta$ of ours is $\frac{P_{\signalL}^\stateL}{P_{\signalL}^\stateL + P_{\signalH}^\stateH} = \frac{0.5}{0.5 + 0.75} = \frac25$. In contrast, in Schoenebeck and Tao's mechanism, the threshold values reported by agents are their posterior estimations about the frequency of signal $\signalL$. These estimates come from agents' beliefs about the true world state, which is either $\Pr[\omega = \stateH \mid S_i = \signalL] = \frac13$ or $\Pr[\omega = \stateH \mid S_i = \signalH] = \frac35$. So the threshold values are either $\Pr[s_j = \signalL \mid S_i = \signalL] =  \frac{1}{3}\times0.25+ \frac{2}{3}\times 0.5 = \frac{5}{12}$ or $\Pr[s_j = \signalL \mid S_i = \signalH] = \frac{3}{5}\times0.25 + \frac{2}{5}\times 0.5 = \frac{7}{20}$. The collective threshold $\delta'$ is the median of all the reported thresholds, so it must lie in the range between $\frac{7}{20}$ and $\frac{5}{12}$.

When the true world state is $L$, the frequency of signal $\signalL$, which is $0.5$, lies above both the threshold $\delta$ and $\delta'$. On the other hand, the frequency of signal $\signalL$, which is $0.25$, is below the thresholds in state $\stateH$. 

The reporting frequency of $\signalL$ has the same properties as long as the proportion of the majority type, $\alpha$, is large enough. However, when the majority proportion is smaller, for example, when $\alpha = 0.84$ and $n\to \infty$, the reporting frequency of $\signalL$ in state $\stateH$ may exceed the threshold in Schoenebeck and Tao's mechanism. That is, when the minority type agents uninformatively claim they receive signal $\signalL$, 
the reporting frequency of $\signalL$, i.e., $\alpha P^\stateH_\signalL + 1-\alpha = 0.84\times 0.25 + 0.16 = 0.37$, may exceed the threshold $\delta'$, which is in the range between $\frac{7}{20}=0.35$ and $\frac{5}{12}\approx0.42$.

On the other hand, our mechanism correctly infers the true world state if $\alpha > \frac45$. In the above case, the reporting rate of signal $\signalL$ in state $\stateL$ and state $\stateH$ are $0.84\times0.5 + 0.16=0.58$ and $0.37$ respectively. Our threshold $\delta = \frac{2}{5}$ succeeds in distinguishing two worlds.
\qed
\end{example}


The following properties of the expected threshold $\delta$ play key roles in our mechanism.

\begin{proposition}\label{claim:frequency}
\begin{equation}
    P_{\signalL}^\stateL > \frac{P_{\signalL}^\stateL}{P_{\signalL}^\stateL + P_{\signalH}^\stateH} > P_{\signalL}^\stateH \text{ and } P_{\signalH}^\stateL < \frac{P_{\signalH}^\stateH}{P_{\signalL}^\stateL + P_{\signalH}^\stateH} < P_{\signalH}^\stateH. \label{ineq: truthful-prob-comparison}
\end{equation}
Furthermore, for $\alpha > \frac1{\Delta + 1}$, we have the following inequalities:
\begin{subequations}
\label{ineq: strategic-prob-comparison}
\begin{align}
    \alpha P_{\signalL}^\stateL &> \frac{P_{\signalL}^\stateL}{P_{\signalL}^\stateL + P_{\signalH}^\stateH} > \alpha P_{\signalL}^\stateH + (1 - \alpha) \\
    \alpha P_{\signalH}^\stateL + (1 - \alpha) &< \frac{P_{\signalH}^\stateH}{P_{\signalL}^\stateL + P_{\signalH}^\stateH} < \alpha P_{\signalH}^\stateH
\end{align}
\end{subequations}
\end{proposition}

The inequality chains in (\ref{ineq: truthful-prob-comparison}) ensure the aggregation of the informed majority decision when all the agents play truthfully. It says that 
the expected fraction of agents receiving signal $\signalL$ in state $\stateL$, which is $P_\signalL^\stateL$, is above our threshold $\delta$; 
that fraction in state $\stateH$, which is $P_\signalL^\stateH$, falls below threshold $\delta$.
These properties ensure the correct state assessment in our mechanism.

The inequality chains in (\ref{ineq: strategic-prob-comparison}) are the key for the truthful-telling profile to be an approximate strong Bayes Nash equilibrium.
Intuitively, it says that, when $\alpha$ is sufficiently large, the minority agents cannot invalidate \eqref{ineq: truthful-prob-comparison} by voting for one alternative deterministically. 

\begin{proof}
    It suffices to demonstrate that $\alpha P_{\signalL}^\stateL > \frac{P_{\signalL}^\stateL}{P_{\signalL}^\stateL + P_{\signalH}^\stateH}$ and $\frac{P_{\signalH}^\stateH}{P_{\signalL}^\stateL + P_{\signalH}^\stateH} < \alpha P_{\signalH}^\stateH$ hold for $1/(\Delta + 1) < \alpha \le 1$. 
    The rest inequalities follow directly from them because each term is 1 minus a corresponding term, for example, $\alpha  P_{\signalL}^\stateH + (1 - \alpha)$ is $1- \alpha P_{\signalH}^\stateH$.

Recall that $\Delta = P_{\signalL}^\stateL - P_{\signalL}^\stateH = P_{\signalH}^\stateH - P_{\signalH}^\stateL$, it can be derived that $P_{\signalL}^\stateL + P_{\signalH}^\stateH = P_{\signalL}^\stateL + P_{\signalH}^\stateL - P_{\signalH}^\stateL + P_{\signalH}^\stateH = 1 + \Delta$.
Thus, we have $\alpha P_{\signalL}^\stateL >\frac{P_{\signalL}^\stateL}{\Delta + 1}= \frac{P_{\signalL}^\stateL}{P_{\signalL}^\stateL + P_{\signalH}^\stateH}$ and $\frac{P_{\signalH}^\stateH}{P_{\signalL}^\stateL + P_{\signalH}^\stateH} =\frac{P_{\signalH}^\stateH}{\Delta + 1} < \alpha P_{\signalH}^\stateH,$
which concludes the proof.
\end{proof}

According to Proposition~\ref{claim:frequency}, when $n\rightarrow\infty$, the mechanism can distinguish between the two possible world states, by comparing the reporting frequency of signal $l$ to the threshold $\delta$. That is to say, if all agents report truthfully, our mechanism aggregates the informed majority decision with probability $1$ as $n \to \infty$. 

More generally, when the number of agents is finite, truthful reporting can still lead to the informed majority decision (with a high probability) in our mechanism.
Although a gap exists between the actual and the expected signal frequencies, this gap, caused by the randomness of signal realization, will exponentially diminish due to the Chernoff bound. 

\begin{lemmarep}
\label{lem: majority wish}
If all agents report truthfully, our mechanism will output the informed majority decision with probability at least $1 - 2\exp(-2c^2n)$. The constant $c$ is defined as \[c = \frac{1}{3}\min\left\{ P_{\signalL}^\stateL - \frac{P_{\signalL}^\stateL}{P_{\signalL}^\stateL + P_{\signalH}^\stateH}, \frac{P_{\signalL}^\stateL}{P_{\signalL}^\stateL + P_{\signalH}^\stateH} -  P_{\signalL}^\stateH \right\}.\]
\end{lemmarep}

We defer the full proof for finite $n$ to the appendix.

\begin{proof}
When all the agents report truthfully, our mechanism can identify the majority type correctly. We only need to show that the true world state can be identified with a high probability.

Suppose the true world state is $\stateL$. The expected fraction of the agents receiving signal $\signalL$ is $P_{\signalL}^\stateL$. By the Chernoff bound, with probability at least $1 - 2\exp(-2c^2n)$, the fraction of agents reporting signal $\signalL$ falls within the interval $[P_\signalL^\stateL - c, P_\signalL^\stateL + c]$, which is larger than the threshold $\hat{\delta} = \frac{P_{\signalL}^\stateL}{P_{\signalL}^\stateL + P_{\signalH}^\stateH}$. The analysis for state $\stateH$ is similar. To sum up, with a probability larger than $1 - 2\exp(-2c^2n)$, the mechanism can correctly identify the world state (i.e., $\hat{\omega} = \omega$), leading to the winning of the majority's preferred alternative.
\end{proof}

\subsection{Analysis for $\alpha > \theta^*$}
Similar to the majority vote mechanism, when the majority proportion is large enough, we find that (approximate) strong equilibria exist: truthful reporting not only can lead to the informed majority decision (with a high probability) but also form an (approximate) strong equilibrium. 



\begin{lemmarep}\label{lem: strong nash equilibrium}
 For configuration $\config$ where the majority proportion \(\alpha > \frac1{\Delta + 1}\) and large enough $n$, truthful reporting forms an $\epsilon$-strong Bayes Nash equilibrium in our mechanism. Here, we define $\epsilon$ as \(\epsilon = \max \{B\exp(-2c^2\lfloor\alpha n\rfloor), 2B^2\exp(-2c^2n)\}\), and define the constant $c$ as \[c = \frac{1}{3}\min\left\{\alpha P_{\signalL}^\stateL - \frac{P_{\signalL}^\stateL}{P_{\signalL}^\stateL + P_{\signalH}^\stateH}, \frac{P_{\signalL}^\stateL}{P_{\signalL}^\stateL + P_{\signalH}^\stateH} - \left[\alpha  P_{\signalL}^\stateH + (1 - \alpha)\right]\right\}.\]
\end{lemmarep}

Below, we prove the above lemma for the case $n\rightarrow\infty$ (in which case $\epsilon=0$).
The full proof with finite $n$ is deferred to the appendix.

According to Lemma~\ref{lem: majority wish}, when $n$ tends to infinity and all agents report truthfully, our mechanism outputs the informed majority decision with probability $1$. Therefore, agents of the majority type achieve their maximum utilities and have no incentive to deviate. Coalition $D$ contains only minority agents. 

We then show the minority agents cannot alter the outcome of our mechanism:
Since the majority type fraction is more than one half ($\alpha > \frac{1}{2}$), minority agents are unable to change the identified majority type $\hat{M}$. Similarly, these agents cannot change the median $\hat{\delta}$ calculated in Step $3$.
The only thing they can influence is the fraction of agents reporting signal $\signalL$. However, in both world states, the deviations in signal fraction cannot cross $\hat{\delta}$.
        \begin{itemize}
            \item In state $\stateL$, even if all minority agents report signal $\signalH$, the fraction reporting signal $\signalL$ will be $\alpha P_{\signalL}^\stateL$.
            By the inequalities in (\ref{ineq: strategic-prob-comparison}), this lower bound of signal $\signalL$'s reporting rate, exceeds the threshold $\hat{\delta} = \frac{P_{\signalL}^\stateL}{P_{\signalL}^\stateL + P_{\signalH}^\stateH}$.  
            \item In state $\stateH$, even if all minority agents report signal $\signalL$, 
            the fraction reporting signal $\signalL$ will be $\alpha P_{\signalL}^\stateH + (1 - \alpha)$. Inequalities in (\ref{ineq: strategic-prob-comparison}) show that this upper bound falls below the threshold.
        \end{itemize} 
        
In both cases, the minority coalition cannot alter the mechanism's output. Thus, there is no beneficial coalition of deviating agents $D$, implying that truth-telling is a strong Nash equilibrium.

\begin{toappendix}
When considering a scenario with a finite agent number $n$, we first need to show that agents of different types are unlikely to cooperate.

\begin{proposition}[No Win-Win Scenario]\label{prop: no win-win}
    Let $\Sigma^*$ be the truthful reporting strategy profile. Let $i_1$ be an arbitrary type-$\typeH$ agent and $i_2$ be an arbitrary type-$\typeL$ agent. For any $\epsilon \ge 2B^2 \exp(-2c^2n)$ and any strategy profile $\Sigma'$, neither of the following conditions holds in our mechanism:
    \begin{itemize}
        \item[(1)] $u_{i_1}(\Sigma') - u_{i_1}(\Sigma^*) > \epsilon\text{~and~}u_{i_2}(\Sigma') - u_{i_2}(\Sigma^*) \ge 0$.
        \item[(2)] $u_{i_2}(\Sigma') - u_{i_2}(\Sigma^*) > \epsilon\text{~and~}u_{i_1}(\Sigma') - u_{i_1}(\Sigma^*) \ge 0$.
    \end{itemize}
\end{proposition}

\begin{proof}[of Proposition~\ref{prop: no win-win}]
    Let us assume type-$\typeH$ is the majority type. By Lemma~\ref{lem: majority wish}, the probability of failing to satisfy the majority agents' wish is less than $2\exp(-2c^2n)$. Thus, any potential increment in the utility for a type-$\typeH$ agent will not surpass $2B\exp(-2c^2n)$, which is less than $\epsilon$. Consequently, Condition (1) is never met.

    Next, we show Condition (2) is also not feasible. 

    Similar to what we have done in the proof of Proposition~\ref{prop: no win-win in majority vote}, we can rewrite agents' ex-ante utility using the ex-post utility. Here we shorten the ex-ante utility notations for agent $i_1, i_2$ as $u_1$ and $u_2$. Analogously, we shorten their ex-post utility notations as $v_1$ and $v_2$. The following two equations hold for arbitrary type-$\typeH$ agent $i_1$ and arbitrary type-$\typeL$ agent $i_2$.
        \begin{equation}
        \Delta u_{1}(\Sigma^*, \Sigma') = \mu(\lambda^\stateH_\accept(\Sigma') - \lambda^\stateH_\accept(\Sigma^*))\Delta v_{1}^\stateH - (1-\mu)(\lambda^\stateL_\accept(\Sigma') -\lambda^\stateL_\accept(\Sigma^*))\Delta v_{1}^\stateL. \label{eq: uH}
    \end{equation}
    \begin{equation}
             \Delta u_{2}(\Sigma^*, \Sigma') = (1-\mu)(\lambda^\stateL_\accept(\Sigma') -\lambda^\stateL_\accept(\Sigma^*))\Delta v_{2}^\stateL - \mu(\lambda^\stateH_\accept(\Sigma') - \lambda^\stateH_\accept(\Sigma^*))\Delta v_{2}^\stateH. \label{eq: uL}
         \end{equation}

    As we can see from the above two equations, for both agents to benefit from the deviation, i.e., $\Delta u_{1}(\Sigma^*, \Sigma') \ge 0$ and $\Delta u_{2}(\Sigma^*, \Sigma') \ge 0$, 
    the term $(\lambda^\stateH_\accept(\Sigma') -\lambda^\stateH_\accept(\Sigma^*))$ and $(\lambda^\stateL_\accept(\Sigma') -\lambda^\stateL_\accept(\Sigma^*))$ must either both be non-negative or both be non-positive.
    
    Intuitively, for type-$\typeH$ agent $i_1$, the deviation results in increasing utility in one state but decreasing utility in the other state. Since truthful reporting already ensures that the informed majority decision is fulfilled with a probability close to $1$, any increment of utility for type-$\typeH$ agent $i_1$ is relatively small. Consequently, the decrement of utility in the other world state is also small. Such deviations do not significantly alter the probability of outputting $\accept$ in either world state, which implies that the impact on the utility of type-$\typeL$ agent $i_2$ is also minor.
   
    Formally, we consider two cases based on the value of $\lambda^\stateH_\accept(\Sigma') - \lambda^\stateH_\accept(\Sigma^*)$: 
    \begin{itemize}
        \item If $\lambda^\stateH_\accept(\Sigma') - \lambda^\stateH_\accept(\Sigma^*) \ge 0$, it holds that $\lambda^\stateL_\accept(\Sigma') - \lambda^\stateL_\accept(\Sigma^*) \ge 0$. By Lemma~\ref{lem: majority wish}, the term $\mu(\lambda^\stateH_\accept(\Sigma') - \lambda^\stateH_\accept(\Sigma^*))$, which is bounded by $\mu(1 - \lambda^\stateH_\accept(\Sigma^*))$, is less than $2\exp(-2c^2n)$. Since $1 \le \Delta v_{1}^\stateL, \Delta v_{1}^\stateH \le B$, According to Equation~(\ref{eq: uH}), for type-$\typeH$ agent $i_1$ to benefit, the value of term $(1-\mu)(\lambda^\stateL_\accept(\Sigma') -\lambda^\stateL_\accept(\Sigma^*))\Delta v_{1}^\stateL$ must not exceed $2B\exp(-2c^2n)$. Combining Equation~(\ref{eq: uL}), we can obtain that \begin{align*}
            \Delta u_{2} &\le (1-\mu)(\lambda^\stateL_\accept(\Sigma') -\lambda^\stateL_\accept(\Sigma^*))\Delta v_{2}^\stateL\\
            &= (1-\mu)(\lambda^\stateL_\accept(\Sigma') -\lambda^\stateL_\accept(\Sigma^*))\Delta v_{1}^\stateL\frac{\Delta v_{2}^\stateL}{\Delta v_{1}^\stateL} \\
            &\le 2B^2 \exp(-2c^2n).
        \end{align*}

        \item If $\lambda^\stateH_\accept(\Sigma') -\lambda^\stateH_\accept(\Sigma^*) < 0$, $\lambda^\stateL_\accept(\Sigma')- \lambda^\stateL_\accept(\Sigma^*) < 0$ holds. Similar to the previous analysis, by Equation~(\ref{eq: uH}), to prevent the utility of type-${\typeH}$ agent $i_1$ from decreasing, it must hold that $$0 > \mu(\lambda^\stateH_\accept(\Sigma') - \lambda^\stateH_\accept(\Sigma^*))\Delta v_{1}^\stateH \ge (1-\mu)(\lambda^\stateL_\accept(\Sigma') -\lambda^\stateL_\accept(\Sigma^*))\Delta v_{1}^\stateL \ge -2B\exp(-2c^2n).$$
        Consequently, by Equation~(\ref{eq: uL}), $$\Delta u_{2} \le -\mu(\lambda^\stateH_\accept(\Sigma') - \lambda^\stateH_\accept(\Sigma^*))\Delta v_{2}^\stateH \le 2B^2\exp(-2c^2n) \le \epsilon.$$
    \end{itemize}
    In both cases, Condition (2) is not satisfied.
\end{proof}
\end{toappendix}

\begin{proof}[of Lemma~\ref{lem: strong nash equilibrium}]
Proposition~\ref{prop: no win-win} demonstrates that, from the truthful-telling profile, a beneficial coalition $D$ contains only minority agents, as no advantageous coalition of both types can exist.
To finish the proof of Lemma~\ref{lem: strong nash equilibrium}, we show that a coalition containing only the minority type agents cannot significantly increase their utility. 

We still assume type-$\typeL$ is the minority type.
This proof follows a similar reasoning to the case when $n \to \infty$. The minority agents cannot alter the identified majority type $\hat{M}$ or the threshold $\hat{\delta}$, they can only deviate the fraction of agents reporting signal $\signalL$.
However, in state $\stateL$, no matter what strategy they use, the expected fraction of signal $\signalL$ reports contributed by the majority type will be $\frac{\lfloor\alpha n\rfloor}{n}P_\stateL^\signalL$. 
Given that $\frac{\lfloor\alpha n\rfloor}{n}P_\stateL^\signalL > (\alpha - \frac{1}{n})P_\stateL^\signalL$, 
the expectation will surpass $\alpha P_\signalL^\stateL - c$ for a large enough $n$. 
According to the Chernoff bound, with a probability of at least $1 - \exp(-2c^2\lfloor\alpha n\rfloor)$, this fraction is greater than $\alpha P_\stateL^\signalL - 2c$, which is above the threshold $\hat{\delta}$. Consequently, the probability of misidentifying the real world state is lower than $\exp(-2c^2\lfloor\alpha n\rfloor)$, suggesting that the potential increment of utility for the minority type is bounded by $B\exp(-2c^2\lfloor\alpha n\rfloor)$.
A similar analysis works for state $\stateH$ where the maximum utility increment for the minority type does not exceed $B\exp(-2c^2\lfloor\alpha n\rfloor)$.

In conclusion, there is no beneficial coalition $D$ for $\epsilon$ value of $$\max\{B\exp(-2c^2\lfloor\alpha n\rfloor), 2B^2\exp(-2c^2n)\}.$$
\end{proof}

\begin{remark}
The above results show that when $\alpha > \frac1{\Delta + 1}$, our mechanism can lead to a ``good'' strong Bayes Nash equilibrium that aggregates the informed majority decision. Not only that, with $\alpha > \frac1{\Delta + 1}$, we can show all the equilibria of our mechanism are ``good'', using the same technique for proving Lemma~\ref{lem: BNE2majoritywish}.

\end{remark}
\subsection{Negative Result for $\alpha \le \theta^\ast$}\label{section: impossibility results}

In this section, we explore the feasibility of inducing a strong Bayes Nash equilibrium to aggregate the informed majority decision when $\alpha \le \theta^\ast$.
Our general impossibility results show that no reasonable mechanisms can meet our requirements.

Inspired by the revelation principle, we first focus on designing truthful mechanisms.
While truthfulness has a clear meaning in our proposed mechanism, it is unclear what truthful means in a general mechanism.
To define the truthfulness of general mechanisms, we assume each mechanism $\mechanism$ provides a mapping $\Pi^\mechanism: \mathcal{T} \times \mathcal{S} \to \mathcal{R}$ that maps agent's private type and signal to the report, instructing how to answer questions. We say agent $i$ play the truthful strategy if $\sigma_i(s) = \Pi^\mechanism(t_i, s)$ for all $s \in \mathcal{S}$. Take the majority vote mechanism as an example, the informative voting strategy can be truthful if we define $\Pi^{\texttt{maj}}(t_i, s)$ as $\accept$ if $t_i = \typeH$ and $s = \signalH$, or $t_i = \typeL$ and $s = \signalL$, and as $\reject$ otherwise. 

Notice that the above definition of truthfulness implicitly implies that the mechanism has to be \emph{ordinal}, in other words, the mechanism does not elicit \emph{cardinal utilities} from agents.
This is because, in the truthful strategy specification $\Pi^\mechanism$, the specified truthful report $\mathcal{R}$ can only depend on $t_i\in\mathcal{T}$ and $s\in\mathcal{S}$ and cannot depend on the utility function $v_i(\cdot,\cdot)$.
Note that our mechanism described in this section is ordinal, and most existing social choice mechanisms are ordinal (if not working on those equilibrium concepts, cardinal preferences are not even defined in most cases).


In the first impossibility result, we examine the existence of truthful anonymous mechanisms, where truthful-telling forms a strong Bayes Nash equilibrium.

\begin{toappendix}
\subsection{Proof Sketch of Theorem~\ref{theorem: impossible result}}
\end{toappendix}

\begin{theoremrep}\label{theorem: impossible result}
For configuration $\config$ where the majority proportion $\alpha \le \frac1{\Delta + 1}$, $0 < \mu < 1, P_\signalH^\stateH = 1/2 + \Delta/2$, and $P_\signalH^\stateL = 1/2 - \Delta/2$, even the number of agents $n \to \infty$, in any anonymous mechanism $\mechanism$, truthful reporting cannot form an $\epsilon$-strong Bayes Nash equilibrium that leads to the informed majority decision with probability $1$. Here, $\epsilon$ is defined as $\epsilon= \frac14\min\{\mu, 1-\mu\}$. 
\end{theoremrep}

To prove this theorem, we construct two environments that are indistinguishable from the mechanisms' view. However, the informed majority decisions in these environments are different, leading to a result that any mechanism will fail to output the informed majority decision in at least one environment. The construction and proof are rather lengthy, so we defer them to the appendix.

\begin{toappendix}
 We try to provide a proof sketch here and defer the formal proof to Section~\ref{section: impossibility proof}.
\begin{proof}[Proof Sketch]
Without loss of generality, let's assume the majority type is type-$\typeL$. We construct two deviations of the minority type and show that at least one of them can bring a utility increment.

Here are some notations involved: let $\pi^T_i = \Pi^\mechanism(t_i, \cdot)$ denote the truth-telling strategy for agent $i$, i.e., the agent answers questions honestly based on his/her knowledge and private signal.
Let $\pi^{\typeL,\signalL} = \Pi^\mechanism(\typeL, \signalL)$ be the strategy that one gives answers as if (s)he was a type-$\typeL$ agent with received signal $\signalL$.

Consider the following two strategy profiles, deviated by type-$\typeH$ agents from the truthful reporting profile:
\begin{itemize}
\item [Profile A:]
\begin{itemize}
    \item All type-$\typeL$ agents report truthfully. Formally, $$\sigma^{(A)}_i = \pi^T_i\text{ for } i\text{ with }t_i = \typeL.$$
    \item Each agent of type-$\typeH$ adopts a strategy where:
    \begin{itemize}
        \item with a probability $q(n) = (1 - \Delta \cdot \lfloor \alpha n\rfloor/\lceil (1 -\alpha)n\rceil)/2$, the agent reports as if she was a type-$\typeL$ agent who received the signal $\signalH$;
        \item with the remaining probability $1 - q(n)$, the agent reports as if she was a type-$\typeL$ agent who received the signal $\signalL$. 
    \end{itemize}
    Formally, $$\sigma^{(A)}_i = \left\{
\begin{array}{rcl}
\pi^{\typeL,\signalH} & & {with~prob~q(n) =(1 - \Delta\cdot\lfloor\alpha n\rfloor/\lceil(1 - \alpha)n\rceil)/2}, \\
\pi^{\typeL,\signalL} & & {otherwise}.
\end{array}
\right.\text{ for } i\text{ with }t_i = \typeH.$$
\end{itemize}

\item [Profile B:] 
\begin{itemize}
    \item Once again, all agents of type-$\typeL$ report truthfully. $$\sigma^{(B)}_i = \pi^T_i~\text{ for } i\text{ with }t_i = \typeL.$$
    \item Each agent of type-$\typeH$ now adopts a strategy where:
    \begin{itemize}
        \item with a probability $1 - q(n)$, she reports as if she was a type-$\typeL$ agent who received the signal $\signalH$.
        \item otherwise, with probability $q(n)$, she reports as if she was a type-$\typeL$ agent who received the signal $\signalL$.
    \end{itemize}

$$\sigma^{(B)}_i = \left\{
\begin{array}{rcl}
\pi^{\typeL,\signalH} & & {with~prob~1 - q(n)}, \\
\pi^{\typeL,\signalL} & & {otherwise}.
\end{array}
\right.\text{ for }i\text{ with }t_i = \typeH.$$
\end{itemize}
\end{itemize}

In both strategies, type-$\typeH$ agents imitate the reporting behavior of type-$\typeL$ agents. The difference between these two strategies lies in the probabilities with which they report having received one signal over the other. 

Combining the two strategy profiles and the two possible world states results in four possible environments. We label them as $A\stateL$, $A\stateH$, $B\stateL$, and $B\stateH$ and focus on the mechanism inputs in environment $A\stateH$ and $B\stateL$. In both environments, the mechanism observes all agents appear to be of type-$\typeL$. The difference is that
\begin{itemize}
    \item In environment $A\stateH$, the count of agents who report receiving signal $\signalH$, denoted by $C_1$, is the sum of two independent binomial distributions:
    \begin{itemize}
        \item $X_1\sim\Bin(\lfloor \alpha \cdot n\rfloor, P_{\signalH}^\stateH = 1/2 + \Delta/2)$, corresponding to the number of actual type-$\typeL$ agents who receive signal $\signalH$.
        \item $Y_1\sim \Bin(\lceil (1-\alpha)\cdot n\rceil, q(n))$, representing the number of type-$\typeH$ agents who report signal $\signalH$ according to the deviation strategy A.
    \end{itemize}
    \item In enrironment $B\stateL$, the number of agents reporting signal $\signalH$, denoted by $C_2$, is the sum of two other independent binomial distributions:
    \begin{itemize}
        \item $X_2\sim \Bin(\lfloor \alpha \cdot n\rfloor, P_{\signalH}^\stateL = 1/2 - \Delta/2)$,
        \item and $Y_2 \sim \Bin(\lceil(1-\alpha)\cdot n\rceil, 1-q(n))$. 
    \end{itemize}
\end{itemize}

We show that any anonymous mechanism $\mechanism$ cannot differentiate between these two environments when the number of agents $n$ is large enough. This inability to distinguish leads to the similarity of the mechanism's outputs, which implies that in at least one environment, type-$\typeH$ agents can benefit from deviation. Therefore, for any mechanism $\mechanism$, truthful reporting is not a strong Bayes Nash equilibrium (even not an approximate strong Bayes Nash equilibrium) succeeding in leading to the informed majority decision.

Detail analysis in formal proof shows that even when $n \to \infty$, for any mechanism, truthful reporting cannot form an $\epsilon$-strong Bayes Nash equilibrium that leads to the informed majority decision, where $\epsilon$ is set to $\epsilon= \frac14\min\{\mu, 1-\mu\}$.
\end{proof}

\subsection{Formal Proof of Theorem~\ref{theorem: impossible result}}
\label{section: impossibility proof}
As aforesaid, we construct two deviation strategies for the minority type agents (without loss of generality, we assume the minority agents are of type-$\typeH$). These strategies lead to two distinct agent profiles, Profile A and Profile B. Combining them with the two possible world states, we obtain four environments: $A\stateL$, $A\stateH$, $B\stateL$, and $B\stateH$.

The only difference between environments $A\stateH$ and $B\stateL$, from an anonymous mechanism's view, lies in the count of agents reporting signal $\signalH$. We denote this count as $C_1$ in environment $A\stateH$ and $C_2$ in environment $B\stateL$. We're now showing that any anonymous mechanism cannot differentiate between these two environments when the number of agents $n$ is large enough.

To quantify this difference, we introduce the concept of total variation distance.

\textbf{Total Variation Distance(TVD):} Let 
$P$ and $Q$ be two probability distributions on a common sample space. The total variation distance $\tvd(P, Q)$ between $P$ and $Q$ is defined as $$\tvd(P, Q) = \sup_{A \in \mathcal{F}} |P(A) - Q(A)|$$ where the supremum is taken over all possible events $A$ in the sample space.\\
For discrete distributions $P$, $Q$, the total variation distance can be written as
    $$\tvd(P, Q) = \frac{1}{2} \sum_{x \in S} |p(x) - q(x)|$$ where $p(x)$ and $q(x)$ are the probability function of $P$ and $Q$, and $S$ is the value range the random variables.\\
And for continuous distributions $P$, $Q$ with density functions $p(x)$ and $q(x)$, the total variation distance can be written as
$$\tvd(P, Q) = \frac{1}{2} \int_{-\infty}^{\infty} |p(x) - q(x)| \, dx.$$

\begin{lemma}[Small TVD between Two Counting Variables]\label{lemma:TVD is small}

For $0 < \alpha \le 1/(\Delta + 1)$, $q(n) =(1 - \Delta\cdot\lfloor\alpha n\rfloor/\lceil(1 - \alpha)n\rceil)/2$, and $C_1, C_2$ corresponding to the counting variables in two environments, the total variation distance between $C_1$ and $C_2$ is at most $o(1)$ when $n \to \infty$.
\end{lemma}
The proof of this lemma will be referred to later. Applying this lemma, we can see for every anonymous mechanism $\mechanism$, the difference in outputs for the two environments is small when $n$ is large enough: 

For a fixed number of agents, let $\mechanism(c)$ be the output of mechanism $\mechanism$ when there are $c$ agents receiving signal $\signalH$. The difference of mechanism outputs in environment $A\stateH$ and $B\stateH$ is bounded by the TVD between $C_1$ and $C_2$.
$$|\Pr_{c\sim C_1}[\mechanism(c) = \accept] - \Pr_{c \sim C_2}[\mechanism(c) = \accept]| \le \tvd(C_1, C_2).$$
Therefore, there exists $n_0$ such that for all $n > n_0$, $$|\Pr_{C_1}[\mechanism(c) = \accept] - \Pr_{ C_2}[\mechanism(c) = \accept]| = |\Pr_{C_1}[\mechanism(c) = \reject] - \Pr_{ C_2}[\mechanism(c) = \reject]| \le 1/4.$$

For these $n$ beyond $n_0$, deviation from truthful reporting can benefit type-$\typeH$ agents in at least one of the two environments: 

In environment $A\stateH$, either $\Pr_{C_1}[\mechanism(c) = \accept] \ge \frac{1}{2}$ or $\Pr_{C_1}[\mechanism(c) = \reject] \ge \frac{1}{2}$ holds. 

\begin{itemize}
    \item When $\Pr_{C_1}[\mechanism(c) = \accept] \ge \frac{1}{2}$, for each agent $i$ with type $t_i = \typeH$, deviation to profile A can bring her a utility 
    $$\begin{aligned}
    &u_i(\sigma_1^{(A)}, \cdots, \sigma_n^{(A)}) \\
    \ge& \mu \Pr_{C_1}[\mechanism(c) = \accept](v_i(\stateH, \accept)-v_i(\stateH, \reject)) +\mu \cdot v_i(\stateH,\reject)+ (1 - \mu)\cdot v_i(\stateL, \accept)\\
    \ge& \mu/2 \cdot [v_i(\stateH, \accept)-v_i(\stateH, \reject)]+\mu \cdot v_i(\stateH,\reject)+ (1 - \mu)\cdot v_i(\stateL, \accept). 
\end{aligned}$$
 \item When $\Pr_{C_1}[\mechanism(c) = \reject] \ge \frac{1}{2}$, $\Pr_{C_2}[\mechanism(c) = \reject]\ge \frac{1}{4}$. For each agent $i$ with $t_i = \typeH$, deviation to profile B can bring her utility
$$\begin{aligned}
    &u_i(\sigma_1^{(B)}, \cdots, \sigma_n^{(B)}) \\
    \ge& (1-\mu) \Pr_{C_2}[\mechanism(c) = \reject](v_i(\stateL, \reject)-v_i(\stateL, \accept)) +\mu \cdot v_i(\stateH,\reject)+ (1 - \mu)\cdot v_i(\stateL, \accept)\\
    \ge& (1-\mu)/4 \cdot [v_i(\stateL, \reject)-v_i(\stateL, \accept)]+\mu \cdot v_i(\stateH,\reject)+ (1 - \mu)\cdot v_i(\stateL, \accept).
\end{aligned}$$    
\end{itemize}

If truthful $\mechanism$ outputs majority wish with probability $1$, then for every agent $i$ with $t_i = \typeH$, truthful reporting gives her utility $$u_i(\sigma^T_1, \cdots, \sigma^T_n) = \mu \cdot v_i(\stateH,\reject)+ (1 - \mu)\cdot v_i(\stateL, \accept).$$ 

No matter in which case, at least one deviation bring an increment about type-$\typeH$ agents' utility of at least $\frac14\min\{\mu, 1-\mu\}$. 
Therefore, truthful reporting is not a 
strong Bayes Nash equilibrium (even not an $\epsilon$-strong Bayes Nash equilibrium with $\epsilon = \frac14\min\{\mu, 1-\mu\}$), conflicting the truthfulness of $\mechanism$. 

\subsubsection{Proof of Lemma~\ref{lemma:TVD is small}}
Now we verify Lemma~\ref{lemma:TVD is small} to finish our proof of the impossible result, whose formal statement is shown below.
\begin{lemma}[Formal Statement of Lemma~\ref{lemma:TVD is small}]\label{formal statement}
Given the conditions \begin{itemize}
    \item $0 < \alpha \le 1/(\Delta + 1)$,
    \item and $q(n) =(1 - \Delta\cdot\lfloor\alpha n\rfloor/\lceil(1 - \alpha)n\rceil)/2$,
\end{itemize} 
     let $C_1$ be the sum of random variables $X_1$ and $Y_1$, let $C_2$ be the sum of random variables$X_2$ and $Y_2$ where
     \begin{itemize}
    \item $X_1$ follows binomial distribution $\Bin(\lfloor \alpha \cdot n\rfloor, 1/2 + \Delta/2)$,
    \item $Y_1$ follows binomial distribution $\Bin(\lceil (1-\alpha)\cdot n\rceil, q)$,
    \item $X_2$ follows $\Bin(\lfloor \alpha \cdot n\rfloor, 1/2 - \Delta/2)$,
    \item $Y_2$ follows $\Bin(\lceil (1-\alpha)\cdot n\rceil, 1-q)$, 
    \item $X_1$, $X_2$, $Y_1$ and $Y_2$ are mutually independent. 
\end{itemize}

As $n$ goes to infinity, the total variation distance between $C_1$ and $C_2$ is $O(\frac{1}{\sqrt{n}})$.
\end{lemma}

To prove this lemma, we use the discretized Gaussian distribution to approximate the binomial distribution and then upper bound the distance between two sums. 

\begin{definition}[Discretized Gaussian]
Let $Z_{\mu, \sigma^2}$ be the discretization of the Gaussian distribution $N(\mu,\sigma^2)$. The probability distribution function of $Z_{\mu, \sigma^2}$ is given as follows: $$\Pr(Z_{\mu, \sigma^2} = i) = \Pr(\frac{i - \mu - 1/2}{\sigma} < Z \le \frac{i - \mu + 1/2}{\sigma})~\forall i \in \mathbb{Z},$$ where $Z$ follows the standard normal distribution $N(0,1)$.
\end{definition}

\begin{proposition}[Distance between Binomial Distribution and Discretized Gaussian~\cite{chen2011normal} ]\label{lemma: Bin and Normal are close}
Let $X$ be a random variable that follows a binomial distribution $\Bin(n, p)$. Define $\mu$ and $\sigma^2$ as the mean and variance of $X$,  where $\mu = np$, $\sigma^2 = np(1-p)$. The total variation distance between the binomial distribution $X$ and its corresponding discretized Gaussian distribution $Z_{\mu, \sigma^2}$
 is at most $\frac{7.6}{\sigma}=O(\frac{1}{\sqrt{n}})$.
\end{proposition}
Proposition~\ref{lemma: Bin and Normal are close} and its proof are provided in literature~\cite{chen2011normal}. 

\begin{proposition}[TV Distance of Sum $\le$ Sum of TV Distances]\label{lemma: sum of TVD is bounded}
Consider four random variables $X_1$, $X_2$, $Y_1$, and $Y_2$ where $X_1$ is independent of $Y_1$, and $X_2$ is independent of $Y_2$. The total variation distance between the sums $X_1 + Y_1$ and $X_2 + Y_2$ is bounded by the combined total variation distance of $X_1$ and $X_2$, and $Y_1$ and $Y_2$. Formally, $$\tvd(X_1 + Y_1, X_2 +Y_2) \le \tvd(X_1, X_2) + \tvd(Y_1, Y_2).$$
\end{proposition}
Proposition~\ref{lemma: sum of TVD is bounded} upper bounds the distance between two sum-distributions by the sum of distances between respective addends.

\begin{proof}[of Proposition~\ref{lemma: sum of TVD is bounded}]
$$\begin{aligned}
    &\tvd(X_1+Y_1, X_2+Y_2) \\
    \le&\frac{1}{2}\sum_{x,y}|P_{X_1,Y_1}(x,y) - P_{X_2,Y_2}(x,y)|\\
    =&\frac{1}{2}\sum_{x,y}|P_{X_1}(x)P_{Y_1}(y) - P_{X_2}(x)P_{Y_2}(y)|\\
    =&\frac{1}{2}\sum_{x,y}|P_{X_1}(x)P_{Y_1}(y) - P_{X_2}(x)P_{Y_1}(y) + P_{X_2}(x)P_{Y_1}(y)- P_{X_2}(x)P_{Y_2}(y)|\\
    \le & \frac{1}{2}\sum_{x,y}| P_{X_1}(x)P_{Y_1}(y) - P_{X_2}(x)P_{Y_1}(y)| + \frac{1}{2}\sum_{x,y}|P_{X_2}(x)P_{Y_1}(y)- P_{X_2}(x)P_{Y_2}(y)|\\
    = & \frac{1}{2}\sum_{x} |P_{X_1}(x) - P_{X_2}(x)|\sum_{y}P_{Y_1}(y) + \frac{1}{2}\sum_{x} P_{X_2}(x)\sum_{y}|P_{Y_1}(y) - P_{Y_2}(y)|\\
    = & \frac{1}{2}\sum_{x} |P_{X_1}(x) - P_{X_2}(x)|+\frac{1}{2}\sum_{y}|P_{Y_1}(y) - P_{Y_2}(y)|\\
    = & \tvd(X_1, X_2) + \tvd(Y_1, Y_2)
\end{aligned}
$$
\end{proof}

\begin{proof}[of Lemma~\ref{formal statement}]
We begin with defining $Z_{X_1}$, a discretized Gaussian distribution corresponding to the binomial random variable $X_1$. The mean $\mu$ and variance $\sigma^2$ of $Z_{X_1}$ are $\mu = \mu_{X_1} = \lfloor\alpha n\rfloor \frac{1 + \Delta}{2}$ and $\sigma^2 = \sigma^2_{X_1} = \lfloor\alpha n\rfloor \frac{1 + \Delta}{2}\cdot\frac{1 - \Delta}{2}$, respectively. According to Proposition~\ref{lemma: Bin and Normal are close}, the total variation distance between $X_1$ and $Z_{X_1}$ is small,  specifically at most $O(\frac{1}{\sqrt{n}})$.

Similarly, we can define $Z_{X_2}$, $Z_{Y_1}$, and $Z_{Y_2}$ for $X_2$, $Y_1$, and $Y_2$. These discretized Gaussian distributions are also very close to their corresponding binomial distributions in terms of total variation distance.

Considering the independence of both $Z_{X_1}$ with $Z_{Y_1}$ and $X_1$ with $Y_1$, we apply Proposition~\ref{lemma: sum of TVD is bounded}. This allows us to bound the total variation distance between the sums $Z_{X_1} + Z_{Y_1}$ and $X_1 + Y_1$ as follows: 
$$\tvd(Z_{X_1} + Z_{Y_1}, X_1+Y_1)\le \tvd(Z_{X_1}, X_1) + \tvd(Z_{Y_1}, Y_1) = O(\frac{1}{\sqrt{n}}).$$ 

In a similar reasoning, the total variation distance between $Z_{X_2} + Z_{Y_2}$ and $X_2+Y_2$ is also $O(\frac{1}{\sqrt{n}})$. 

Since the triangle inequality holds for total variation distance, now we only need to show the total variation distance between $Z_{X_1} + Z_{Y_1}$ and $Z_{X_2} + Z_{Y_2}$ is small.
Note that the mean values of $Z_{X_1}$ and $Z_{X_2}$ differ by $\lfloor \alpha n\rfloor \cdot \Delta$, which is also the mean value difference between $Z_{Y_1}$ and $Z_{Y_2}$. This difference can be decomposed into an integer part $k$ and a decimal part $r, 0\le r < 1$.

Adjusting $Z_{X_2}$ and $Z_{Y_2}$ with $k$, we can get $Z'_{X_2} = Z_{X_2} + k$ and $Z'_{Y_2} = Z_{Y_2} - k$, which have the same summation as $Z_{X_2}$ and $Z_{Y_2}$: $$Z'_{X_2} + Z'_{Y_2} = Z_{X_2} + Z_{Y_2}.$$ 
If $Z'_{X_2}$ is in $O(\frac{1}{\sqrt{n}})$ distance with $Z_{X_1}$, and $Z'_{Y_2}$ is in $O(\frac{1}{\sqrt{n}})$ distance with $Z_{Y_1}$, using Proposition~\ref{lemma: sum of TVD is bounded} we have $$\begin{aligned}
    \tvd(Z_{X_1}+Z_{Y_1}, Z_{X_2} + Z_{Y_2}) &= \tvd(Z_{X_1}+Z_{Y_1}, Z'_{X_2} + Z'_{Y_2})\\& \le \tvd(Z_{X_1}, Z'_{X_2}) + \tvd(Z_{Y_1}, Z'_{Y_2})\\& = O(\frac{1}{\sqrt{n}}).
\end{aligned}$$
Therefore,
$$\begin{aligned}    
&~d_{TV}(X_1+Y_1, X_2 + Y_2)\\
\le &~d_{TV}(X_1+Y_1,Z_{X_1} + Z_{Y_1}) + d_{TV}(Z_{X_1} + Z_{Y_1}, Z_{X_2} + Z_{Y_2}) + d_{TV}(Z_{X_2} + Z_{Y_2}, X_2 + Y_2) \\= &~O(\frac{1}{\sqrt{n}}).
\end{aligned}$$

The remaining thing is to show both $\tvd(Z_{X_1}, Z'_{X_2})$ and $\tvd(Z_{Y_1}, Z'_{Y_2})$ are $O(\frac{1}{\sqrt{n}})$.

The proximity of $Z_{X_1}$ and $Z'_{X_2}$ can be found from their definitions:
$$\Pr[Z_{X_1} = i] = \Pr(i - 1/2 < \sigma_{X_1} * Z + \mu_{X_1}\le i + 1/2),$$and $$\Pr[Z'_{X_2} = i] = \Pr(i - 1/2 < \sigma_{X_2} * Z + \mu_{X_2} + k\le i + 1/2).$$
Since $\sigma_{X_1} = \sigma_{X_2}, \mu_{X_1} = \mu_{X_2} + k + r$, 
$$\Pr[Z'_{X_2} = i] = \Pr(i - 1/2 < \sigma_{X_1} * (Z - \frac{r}{\sigma_{X_1}}) + \mu_{X_1}\le i + 1/2).$$
Thus, 
\begin{align*}
    \tvd(Z_{X_1}, Z'_{X_2}) & \le \tvd(Z, Z - \frac{r}{\sigma_{X_1}})\\
    & = \Phi(\frac{r}{2\sigma_{X_1}}) - \Phi(-\frac{r}{2\sigma_{X_1}})\tag{$\Phi(x)$ is the cumulative distribution function of $N(0,1)$}\\
    &= \int_{-r/2\sigma_{X_1}}^{r/2\sigma_{X_1}} \frac{1}{\sqrt{2\pi}}\exp(-x^2/2) \, dx\\
    &\le \frac{r}{\sigma_{X_1}} \cdot \frac{1}{\sqrt{2\pi}} \tag{$\exp(-x^2/2) \le 1$ holds for all $x$}\\
    &= O(\frac{1}{\sqrt{n}}).
\end{align*}

A similar argument applies to $Z_{Y_1}$ and $ Z'_{Y_2}$, where $\tvd(Z_{Y_1}, Z'_{Y_2}) \le O(\frac{1}{\sqrt{n}})$. 
\end{proof}
\end{toappendix}

\begin{toappendix}
\subsection{Proof of Theorem~\ref{theorem: general impossibility}}
\end{toappendix}

Moreover, the impossibility of informed majority decision aggregation is not restricted to truthful mechanisms. In the following theorem, we consider any \emph{symmetric strong Bayes Nash equilibrium}. We demonstrate that no anonymous mechanism can lead to a symmetric strong Bayes Nash equilibrium that ensures the informed majority decision.

\begin{theoremrep}\label{theorem: general impossibility}
For configuration $\config$ where the majority proportion $\alpha \le \frac1{\Delta + 1}$, $0 < \mu < 1, P_\signalH^\stateH = 1/2 + \Delta/2$, and $P_\signalH^\stateL = 1/2 - \Delta/2$, even the number of agents $n \to \infty$, in any anonymous mechanism $\mechanism$, no symmetric $\epsilon$-strong Bayes Nash equilibrium can guarantee the informed majority decision with probability $1$. Here, we define $\epsilon$ as $\epsilon= \frac14\min\{\mu, 1-\mu\}$.
\end{theoremrep}

A symmetric strategy profile where all agents play the strategy $\sigma:\mathcal{T}\times\mathcal{S}\to\Delta(\mathcal{R})$ can be characterized by the tuple $\{\sigma^\typeL, \sigma^\typeH\}$ with $\sigma^\typeL(\cdot)=\sigma(\typeL,\cdot)$ and $\sigma^\typeH(\cdot)=\sigma(\typeH,\cdot)$, where all type-$\typeL$ agents uniformly adopt strategy $\sigma^\typeL$ and all type-$\typeH$ agents take strategy $\sigma^\typeH$.
In the proof of Theorem~\ref{theorem: general impossibility}, we will use this alternative interpretation of symmetric strategies.

The proof of this theorem is built upon Theorem~\ref{theorem: impossible result}. By employing a technique similar to the revelation principle, we demonstrate that a symmetric strong Bayes Nash equilibrium, which leads to the informed majority decision, suggests a truth-telling equilibrium in a new mechanism that similarly ensures the informed majority decision. Consequently, such symmetric strong Bayes Nash equilibria cannot exist.

\begin{proof}
     The proof of this theorem is supported by Theorem~\ref{theorem: impossible result}. For contradiction, suppose there exists a mechanism $\mechanism$ and a symmetric $\epsilon$-strong Bayes Nash equilibrium $\{\sigma^\typeL, \sigma^\typeH\}$ that consistently yields the informed majority decision. 
     
     We challenge this by considering a new mechanism $\mechanism'$ that requests agents' preference types and private signals. Similar to the revelation principle, the new mechanism simulates the equilibrium strategy for the agents. For a type-$\typeL$ agent receiving signal $s$, $\mechanism'$ simulates his/her report according to the distribution $\sigma^\typeL(s)$. For a type-$\typeH$ agent, the mechanism emulates $\sigma^\typeH(s)$. After generating the report profile $(r_1, r_2, \ldots, r_n)\in \mathcal{R}^n$, $\mechanism'$ runs $\mechanism$ and adopts its decision.

    Then we examine the truthfulness of mechanism $\mechanism'$. Given that $\{\sigma^\typeL, \sigma^\typeH\}$ ensures the informed majority decision under $\mechanism$, truthful reporting under $\mechanism'$ should also lead to the informed majority decision.
    For the majority-type agents, any deviation on the outcome lowers their utility, thus they will not join a deviation coalition. For the minority-type agents, any effective deviation from truthful reporting implies a significantly advantageous deviation from $\{\sigma^\typeL, \sigma^\typeH\}$ under the original mechanism $\mechanism$, which is not possible because strategies $\{\sigma^\typeL, \sigma^\typeH\}$ is an $\epsilon$-strong Bayes Nash equilibrium. Therefore, no coalition can benefit from deviating truthful strategy.
    Truthful reporting forms an $\epsilon$-strong Bayes Nash equilibrium under $\mechanism'$. 
    
    However, this contradicts the understanding in Theorem~\ref{theorem: impossible result} that no truthful mechanism can consistently secure the informed majority decision for the assumed configuration, thereby challenging the assumption.
\end{proof}


\begin{toappendix}    

\section{Elicitation of $\frac{P_\signalL^\stateL}{P_\signalL^\stateL + P_\signalH^\stateH}$}\label{sect: threshold elicitation}

Our mechanism requires agents to accurately report $\frac{P_\signalL^\stateL}{P_\signalL^\stateL + P_\signalH^\stateH}$. This is challenging to implement in practical scenarios like political forecasting, where extracting the conditional probability of signals across different world states is intricate.

To address this complexity, we develop a more user-friendly questionnaire to implement our mechanism. This questionnaire gathers agents' predictions of others and their confidence in the world state, which has been proven to be pragmatic in previous research \cite{prelec2017solution}, and enables us to derive the value of $\frac{P_\signalL^\stateL}{P_\signalL^\stateL + P_\signalH^\stateH}$. 

Here is our questionnaire:
\begin{itemize}
    \item[(1)] Choose one of the followings:
    \begin{itemize}
        \item [(a)]I prefer alternative $\accept$ when the world state is $\stateL$.
        \item [(b)]I prefer alternative $\reject$ when the world state is $\stateL$.
    \end{itemize}
    \item[(2)] Choose one of the followings:
       \begin{itemize}
        \item [(a)]My received signal is $\signalL$. 
        \item [(b)]My received signal is $\signalH$. 
    \end{itemize}
    \item[(3)]What percentage of the participants do you believe will choose option (a) in the second question?
    \item[(4)]
    How likely do you think the state of the world will be $\stateL$?
    \item[(5)]Suppose you received the opposite signal, please answer the third and the fourth questions again.
\end{itemize}

The first two questions in the provided questionnaire are designed to gather information about participants' preference types and their private signals. The third question seeks to elicit participants' predictions regarding the signals of others, namely, $\Pr[S_j = \signalL \mid S_i = s]$, where $i \neq j$ and $s\in \{\signalL, \signalH\}$. 
The fourth question asks participants' confidence in the state being ``$\stateL$'', corresponding to $\Pr[\omega = \stateL \mid S_i = s], s \in \{\signalL,\signalH\}$.
In the followings, we show that the information gathered above is sufficient to recover the value of $\frac{P_\signalL^\stateL}{P_\signalL^\stateL + P_\signalH^\stateH}$. 

Relationships between these extracted probabilities are described as follows:
\begin{subequations}
\begin{align}
    \Pr[S_j = \signalL \mid S_i = \signalL] = P_\signalL^\stateH + P[\omega = \stateL \mid S_i = \signalL](P_\signalL^\stateL - P_\signalL^\stateH),\label{eq: PeerPrediction1}\\
    \Pr[S_j = \signalL \mid S_i = \signalH] = P_\signalL^\stateH + P[\omega = \stateL \mid S_i = \signalH](P_\signalL^\stateL - P_\signalL^\stateH).\label{eq: PeerPrediction2}
\end{align}
\end{subequations}

These two equations allow us to calculate the values of interest:
\begin{equation}\label{eq: signal difference}
    \Delta = P_\signalL^\stateL - P_\signalL^\stateH = \frac{\Pr[S_j = \signalL \mid S_i = \signalL] - \Pr[S_j = \signalL \mid S_i = \signalH]}{\Pr[\omega = \stateL \mid S_i = \signalL] - \Pr[\omega = \stateL \mid S_i = \signalH]},
\end{equation}
\begin{equation}\label{eq: P_l_H}
   P_\signalL^\stateH = \Pr[S_j = \signalL \mid S_i = \signalL] - \Pr[\omega = \stateL \mid S_i = \signalL] \cdot (P_\signalL^\stateL - P_\signalL^\stateH), 
\end{equation}
\begin{equation}\label{eq: P_l_L}
    P_\signalL^\stateL = \Pr[S_j = \signalL \mid S_i = \signalL] + (1 - \Pr[\omega = \stateL \mid S_i = \signalL]) \cdot (P_\signalL^\stateL - P_\signalL^\stateH).
\end{equation}
Equation (\ref{eq: signal difference}) is obtained by subtracting (\ref{eq: PeerPrediction2}) from (\ref{eq: PeerPrediction1}), extracting the difference in signal probabilities between two states. Equation (\ref{eq: P_l_H}) follows directly from (\ref{eq: PeerPrediction1}), and Equation (\ref{eq: P_l_L}) can be derived from the following equation $$\Pr[S_j = \signalL \mid S_i = \signalL] = P_\signalL^\stateL + \Pr[\omega = \stateH \mid S_i = \signalL](P_\signalL^\stateH - P_\signalL^\stateL) = P_\signalL^\stateL - (1-\Pr[\omega = \stateL \mid S_i = \signalL])(P_\signalL^\stateL - P_\signalL^\stateH) .$$

Finally, we can derive the value of $\frac{P_\signalL^\stateL}{P_\signalL^\stateL + P_\signalH^\stateH}$ by 
$$\frac{P_\signalL^\stateL}{P_\signalL^\stateL + P_\signalH^\stateH} = \frac{P_\signalL^\stateL}{P_\signalL^\stateL + 1 - P_\signalL^\stateH} = \frac{\Pr[S_j = \signalL \mid S_i = \signalL] + (1 - \Pr[\omega = \stateL \mid S_i = \signalL]) \cdot (P_\signalL^\stateL - P_\signalL^\stateH)}{1 + P_\signalL^\stateL - P_\signalL^\stateH}.$$

\end{toappendix}
\section{Comparing Two Mechanisms}\label{section: connection}
In this section, we compare our mechanism with the majority vote mechanism. 
Both the majority vote and our truthful mechanism utilize the difference in signal frequencies across states, $\Delta$, to distinguish the true state or the informed majority decision. 
However, they do this in different ways and result in different thresholds, $\thetaMajority$ and $\theta^\ast$, for the majority proportion.

\textbf{Different Methods of Leveraging Signal Differences}. In our mechanism, the ratio $\frac{P_\signalL^\stateL}{P_\signalL^\stateL + P_\signalH^\stateH}$ serves as a threshold and the reporting signal frequencies of $\signalL$ in the two states naturally lie at both sides of the threshold.
In particular, we have shown in Proposition~\ref{claim:frequency} that the fraction of agents receiving signal $\signalL$ is below $\frac{P_\signalL^\stateL}{P_\signalL^\stateL + P_\signalH^\stateH}$ in state $\stateH$, and above $\frac{P_\signalL^\stateL}{P_\signalL^\stateL + P_\signalH^\stateH}$ in state $\stateL$. 
Our mechanism uses this observation to identify the true state.

In the majority vote mechanism, the midpoint, $\frac12$, is the threshold of the vote share of alternative $\accept$. 
When the signals are unbiased, i.e., with $P_\signalL^\stateH<\frac12<P_\signalL^\stateL$, the informed majority decision wins in both states as long as the majority agents vote according to their received signals. 
When the signals are biased, with one signal occurring more often than the other in both states, the majority-type agents need to strategically ``shift'' the signal frequencies $P_\signalL^\stateH$ and $P_\signalL^\stateL$ by adjusting their voting strategies. 
By making the adjusted frequencies on the opposite sides of $\frac12$, the majority-type agents ensure their informed decision stands out. 

\textbf{Different Thresholds for Majority Proportion}. Two mechanisms induce different majority proportion thresholds for the strong Bayes Nash equilibrium. 
The threshold is $\frac{1}{2M}$ in the majority vote mechanism and $\frac1{\Delta + 1}$ in our mechanism. Notably, as we verify in Proposition~\ref{obs: threshold ineq}, the majority vote has a stronger requirement to attain equilibrium because $\frac{1}{2M} \ge \frac1{\Delta + 1}$ always holds.

For the existence of strong Bayes Nash equilibrium, which is equivalent to the perfect elicitation of the informed majority decision, the proportion of minority agents should be small so they cannot confuse the mechanism's judgment about the world state. 
Intuitively, the larger the difference between the two states, the harder the minority can confuse the mechanism. 
This difference corresponds to the difference in signal frequencies, $\Delta$, in our mechanism. 
On the other hand, in the majority vote, it becomes the difference in the expected vote share of $\accept$ under the majority type's optimal strategy, whose value is $p_\accept^\stateH(\delta_\signalL^*, \delta_\signalH^*) - p_\accept^\stateL(\delta_\signalL^*, \delta_\signalH^*) = \Delta_\accept^\stateH(\delta_\signalL^*, \delta_\signalH^*) -\Delta_\accept^\stateL(\delta_\signalL^*, \delta_\signalH^*) = \Delta\cdot (\delta_\signalH^* + \delta_\signalL^*)$.
Given that the optimal strategy satisfies $\delta_\signalH^*, \delta_\signalL^* \in (0, \frac12]$, the differential in majority voting, $\Delta\cdot (\delta_\signalH^* + \delta_\signalL^*)$, cannot surpass $\Delta$. 
This leads to a higher threshold on the majority proportion in the majority vote mechanism.

In conclusion, our mechanism ``places'' a threshold $\frac{P_\signalL^\stateL}{P_\signalL^\stateL + P_\signalH^\stateH}$ between $P_\signalL^\stateH$ and $P_\signalL^\stateL$, whereas, in the majority vote mechanisms, agents strategically ``shift'' the values of $P_\signalL^\stateH$ and $P_\signalL^\stateL$ to make them lie on different side of $\frac12$.
The shifting process reduces the margin $\Delta$ by a factor of $(\delta_\signalH^* + \delta_\signalL^*)$, which leaves more room for the minority type agents to flip the outcome of the election.


\section{Conclusion and Future Work}
In this paper, we studied the voting game with two alternatives $\{\accept, \reject\}$ and coalitional agents, where there are two groups of voters: type-$\typeL$ voters with $v_i(\stateL, \accept) > v_i(\stateL, \reject)\text{ and }v_i(\stateH, \reject) > v_i(\stateH, \accept)$ and type-$\typeH$ voters with $v_i(\stateL, \reject) > v_i(\stateL, \accept)\text{ and }v_i(\stateH, \accept) > v_i(\stateH, \reject)$. We identified sharp thresholds for the fraction of the majority-type voters necessary for reaching the informed majority decision in strong Nash equilibrium.

One natural future work is to extend the characterization to the setting with general voters' utility functions $v_i(\cdot,\cdot)$.
Notice that, for general utility functions, voters can be classified into four types.
Other than the two ``contingent'' types (Type-$\typeH$ and Type-$\typeL$) studied in this paper, there are two other ``predetermined'' types of voters:
\begin{itemize}
    \item Type-$\accept$ voters who always prefer the alternative $\accept$: $v_i(\stateL, \accept) > v_i(\stateL, \reject)\text{ and }v_i(\stateH, \accept) > v_i(\stateH, \reject)$, and
    \item Type-$\reject$ voters who always prefer the alternative $\reject$: $v_i(\stateL, \reject) > v_i(\stateL, \accept)\text{ and }v_i(\stateH, \reject) > v_i(\stateH, \accept)$.
\end{itemize}

Han et al.'s work~\cite{han2023wisdom} involves three types of voters (Type-$\typeH$, Type-$\accept$, and Type-$\reject$) and some assumptions on voters' cardinal utilities. In their preference setting, they show any strong Bayes Nash equilibrium under the majority voting rule can achieve the informed majority decision. For general cases with all four types of voters, how to characterize the condition for achieving the informed majority decision? This may require combining the results in Han et al.~\cite{han2023wisdom} and the observations in this paper.

A further extension is to study the setting with more than two world states or more than two alternatives.
Under this setting, the meaning of ``antagonistic preferences'' becomes unclear, and it is more natural to consider general utility functions $v_i:\Omega\times\mathbb{A}\to\{0,1,\ldots,B\}$ where $\Omega$ is no longer binary and $\mathbb{A}=\{\accept_1,\ldots,\accept_m\}$ is the set of $m$ alternatives.
In this case, voters can be classified into $(m!)^{|\Omega|}$ types: one preference ranking over the $m$ alternatives for each world state.
We believe general characterizations for achieving the informed majority decision is worth exploring in cases with $m = 2$ and non-binary world states.
For non-binary alternatives, however, even in the special case when $m=3$ and all the voters are predetermined (whose rankings over the alternatives do not depend on the world state), the Gibbard-Satterthwaite Theorem~\cite{gibbard1973manipulation, satterthwaite1975strategy} implies that no meaningful aggregation of the preferences is possible if voters are strategic.

\begin{credits}
\subsubsection{\ackname} 
This study was supported by the National Natural Science Foundation of China (No. 62172012 and No. 62102252).

\end{credits}

\bibliographystyle{splncs04}
\bibliography{reference}

\end{document}